\def\isarxivversion{1} 

\ifdefined\isarxivversion
\documentclass[11pt]{article}
\else
\documentclass{article}
\usepackage[final]{neurips_2019}
\fi

\usepackage[utf8]{inputenc} 
\usepackage[T1]{fontenc}    
\usepackage{hyperref}       
\usepackage{url}            
\usepackage{booktabs}       
\usepackage{amsfonts}       
\usepackage{nicefrac}       
\usepackage{microtype}      

\usepackage{amsmath}
\usepackage{amsthm}
\usepackage{amssymb}
\usepackage{algorithm}
\usepackage{subfig}
\usepackage{color}
\usepackage[english]{babel}
\usepackage{graphicx}
\usepackage{wrapfig,epsfig}
\usepackage{epstopdf}
\usepackage{url}
\usepackage{color}
\usepackage{epstopdf}
\usepackage{algpseudocode}
\usepackage{scrextend}
\usepackage[T1]{fontenc}
\usepackage{bbm}
\usepackage{comment}
\hypersetup{colorlinks=true,citecolor=blue,linkcolor=blue}

\graphicspath{{./figs/}}

\ifdefined\isarxivversion
\usepackage[margin=1in]{geometry}
\else 

\fi

\ifdefined\isarxivversion

\author{
	Zhao Song\thanks{\texttt{zhaosong@uw.edu}. University of Washington.}
	\and
	David P. Woodruff\thanks{\texttt{dwoodruf@cs.cmu.edu}. Carnegie Mellon University.}
	\and
	Peilin Zhong\thanks{\texttt{pz2225@columbia.edu}. Columbia University.}
}
\else

\author{
	Zhao Song\thanks{equal contribution.}\\
	University of Washington \\
	\texttt{magic.linuxkde@gmail.com}
	\And
	 David P. Woodruff\printfnsymbol{1}  \\
	Carnegie Mellon University\\
	\texttt{dwoodruf@cs.cmu.edu}
	\And
	Peilin Zhong\printfnsymbol{1} \\
	Columbia University\\
	\texttt{pz2225@columbia.edu}
}
\fi

\date{}

\newtheorem{theorem}{Theorem}[section]
\newtheorem{lemma}[theorem]{Lemma}
\newtheorem{definition}[theorem]{Definition}

\newtheorem{fact}[theorem]{Fact}

\newtheorem{claim}[theorem]{Claim}

\newtheorem{question}[theorem]{Question}

\newcommand{\wh}{\widehat}
\newcommand{\wt}{\widetilde}

\newcommand{\eps}{\epsilon}

\newcommand{\R}{\mathbb{R}}

\renewcommand{\varepsilon}{\epsilon}
\renewcommand{\tilde}{\wt}

\DeclareMathOperator*{\E}{{\bf {E}}}
\DeclareMathOperator*{\argmin}{argmin}

\DeclareMathOperator{\OPT}{OPT}

\DeclareMathOperator{\poly}{poly}

\DeclareMathOperator{\nnz}{nnz}

\DeclareMathOperator{\rank}{rank}

\DeclareMathOperator{\sign}{sign}

\DeclareMathOperator{\ati}{ati}
\DeclareMathOperator{\sym}{mon}
\DeclareMathOperator{\mon}{mon}
\DeclareMathOperator{\relu}{ReLU}
\DeclareMathOperator{\reg}{reg}
\DeclareMathOperator{\group}{group}

\makeatletter
\newcommand*{\RN}[1]{\expandafter\@slowromancap\romannumeral #1@}
\newcommand{\printfnsymbol}[1]{%
	\textsuperscript{\@fnsymbol{#1}}%
}
\makeatother

\ifdefined\isarxivversion

\title{Towards a Zero-One Law for Column Subset Selection\thanks{A preliminary version of this paper appears in Proceedings of Thirty-third Conference on Neural Information Processing Systems (NeurIPS 2019).}}

\else

\title{Towards a Zero-One Law for Column Subset Selection}

\fi

\begin{document}

\ifdefined\isarxivversion

\begin{titlepage}
\maketitle
\begin{abstract}
There are a number of approximation algorithms for NP-hard versions of low rank approximation, such as finding a rank-$k$ matrix $B$ minimizing the sum of absolute values of differences to a given $n$-by-$n$ matrix $A$, $\min_{\textrm{rank-}k~B}\|A-B\|_1$, or more generally finding a rank-$k$ matrix $B$ which minimizes the sum of $p$-th powers of absolute values of differences, $\min_{\textrm{rank-}k~B}\|A-B\|_p^p$. Many of these algorithms are linear time columns subset selection algorithms, 
returning a subset of $\poly(k \log n)$ columns whose cost is no more than a $\poly(k)$ factor larger than the cost of the best rank-$k$ matrix. 
The above error measures are special cases of the following general entrywise
low rank approximation problem: given an arbitrary function $g:\mathbb{R} \rightarrow \mathbb{R}_{\geq 0}$, find a rank-$k$ matrix $B$ which minimizes $\|A-B\|_g = \sum_{i,j}g(A_{i,j}-B_{i,j})$. A natural question is which functions $g$ admit efficient approximation algorithms? Indeed, this is a central question of recent work studying generalized low rank models. In this work we give approximation algorithms for {\it every} function $g$ which is approximately monotone and satisfies an approximate triangle inequality, and we show both of these conditions are necessary. Further, our algorithm is efficient if the function $g$ admits an efficient approximate regression algorithm. Our approximation algorithms handle functions which are not even scale-invariant, such as the Huber loss function, which we show have very different structural properties than $\ell_p$-norms, e.g., one can show the lack of scale-invariance causes any column subset selection algorithm to provably require a $\sqrt{\log n}$ factor larger number of columns than $\ell_p$-norms; nevertheless we design the first efficient column subset selection algorithms for such error measures.

\end{abstract}
\thispagestyle{empty}
\end{titlepage}

\else

\maketitle
\begin{abstract}
There are a number of approximation algorithms for NP-hard versions of low rank approximation, such as finding a rank-$k$ matrix $B$ minimizing the sum of absolute values of differences to a given $n$-by-$n$ matrix $A$, $\min_{\textrm{rank-}k~B}\|A-B\|_1$, or more generally finding a rank-$k$ matrix $B$ which minimizes the sum of $p$-th powers of absolute values of differences, $\min_{\textrm{rank-}k~B}\|A-B\|_p^p$. Many of these algorithms are linear time columns subset selection algorithms, 
returning a subset of $\poly(k \log n)$ columns whose cost is no more than a $\poly(k)$ factor larger than the cost of the best rank-$k$ matrix. 
The above error measures are special cases of the following general entrywise
low rank approximation problem: given an arbitrary function $g:\mathbb{R} \rightarrow \mathbb{R}_{\geq 0}$, find a rank-$k$ matrix $B$ which minimizes $\|A-B\|_g = \sum_{i,j}g(A_{i,j}-B_{i,j})$. A natural question is which functions $g$ admit efficient approximation algorithms? Indeed, this is a central question of recent work studying generalized low rank models. In this work we give approximation algorithms for {\it every} function $g$ which is approximately monotone and satisfies an approximate triangle inequality, and we show both of these conditions are necessary. Further, our algorithm is efficient if the function $g$ admits an efficient approximate regression algorithm. Our approximation algorithms handle functions which are not even scale-invariant, such as the Huber loss function, which we show have very different structural properties than $\ell_p$-norms, e.g., one can show the lack of scale-invariance causes any column subset selection algorithm to provably require a $\sqrt{\log n}$ factor larger number of columns than $\ell_p$-norms; nevertheless we design the first efficient column subset selection algorithms for such error measures.

\end{abstract}

\fi

\section{Introduction}
A well-studied problem in machine learning and numerical linear algebra, with applications to recommendation systems, text mining, and computer vision, is that of computing a low-rank approximation of a matrix. Such approximations reveal low-dimensional structure, provide a compact way of storing a matrix, and can quickly be applied to a vector.

A commonly used version of the problem is to compute a near optimal low-rank approximation with respect to the Frobenius norm. That is, given an $n \times n$ input matrix $A$ and an accuracy parameter $\epsilon > 0$, output a rank-$k$ matrix $B$ with large probability so that 
$\|A-B\|_F^2 \leq (1+\epsilon) \|A-A_k\|_F^2,$ where for a matrix $C$, $\|C\|_F^2 = \sum_{i,j} C_{i,j}^2$ is its squared Frobenius norm, and $A_k =\textrm{argmin}_{\textrm{rank-}k~B}\|A-B\|_F$. $A_k$ can be computed exactly using the singular value decomposition (SVD), but takes $O(n^3)$ time in practice and $n^{\omega}$ time in theory, where $\omega \approx 2.373$ is the exponent of matrix multiplication \cite{s69,cw87,w12,lg14}.

S\'arlos \cite{s06} showed how to achieve the above guarantee with constant probability in $\tilde O(\nnz(A) \cdot k/\epsilon) + n \cdot \poly(k/\epsilon)$ time, where $\nnz(A)$ denotes the number of non-zero entries of $A$. This was improved in \cite{cw13,mm13,nn13,bdn15,c16} using sparse random projections in $O(\nnz(A)) + n \cdot \poly(k/\epsilon)$ time. Large sparse datasets in recommendation systems are common, such as the Bookcrossing ($100K \times 300K$ with $10^6$ observations) \cite{zmkl05} and Yelp datasets ($40K \times 10K$ with $10^5$ observations) \cite{y14}, and this is a substantial improvement over the SVD. 

To understand the role of the Frobenius norm in the algorithms above, we recall a standard motivation for this error measure. 
Suppose one has $n$ data points in a $k$-dimensional subspace of $\mathbb{R}^d$, where $k \ll d$. We can write these points as the rows of an $n \times d$ matrix
$A^*$ which has rank $k$. The matrix $A^*$ is often called the {\it ground truth matrix}. In a number of settings,
due to measurement noise or other kinds of noise,  
we only observe the matrix $A=A^*+\Delta$, where each entry of the {\it noise matrix} $\Delta\in\mathbb{R}^{n\times n}$ is an i.i.d. 
random variable from a certain mean-zero noise distribution $\mathcal{D}$. 
One method for approximately recovering $A^*$ from $A$ is maximum likelihood estimation.
Here one tries to find a matrix $B$ maximizing the log-likelihood:
$\max_{\textrm{rank-}k~B} \sum_{i,j} \log p(A_{i,j}-B_{i,j}),$
where $p(\cdot)$ is the probability density function of the underlying noise distribution $\mathcal{D}.$
For example, when the noise distribution is Gaussian with mean zero and variance $\sigma^2$, denoted by $N(0,\sigma^2),$ 
then the optimization problem is
$
\max_{\textrm{rank-}k~B} \sum_{i,j} \left(\log(1/\sqrt{2\pi\sigma^2})-(A_{i,j}-B_{i,j})^2/(2\sigma^2)\right),
$
which is equivalent to solving the Frobenius norm loss low rank approximation problem defined above.

The Frobenius norm loss, while having nice statistical properties for Gaussian noise, is well-known to be sensitive to outliers. 
Applying the same maximum likelihood framework above to other kinds of noise distributions results in minimizing other kinds of
loss functions. 
In general, if the density function of the underlying noise $\mathcal{D}$ is 
$
p(z) = c\cdot e^{-g(z)},
$
where $c$ is a normalization constant, then the maximum likelihood estimation problem for this noise 
distribution becomes the following generalized entry-wise loss low rank approximation problem:
$
\min_{\textrm{rank-}k~B} \sum_{i,j} g(A_{i,j}-B_{i,j}) = \min_{\textrm{rank-}k~B}\|A-B\|_g,
$
which is a central topic of recent work on {\it generalized low-rank models} \cite{uhzb16}. 
For example, when the noise is Laplacian, the entrywise $\ell_1$ loss is the maximum likelihood estimation, which is 
also robust to sparse outliers. A natural setting is when the noise is a 
mixture of small Gaussian noise and sparse outliers; this noise distribution is referred to as the {\it Huber density}. 
In this case the Huber loss function gives the maximum likelihood estimate~\cite{uhzb16}, where the Huber function \cite{h64} is defined 
to be: $g(x) = x^2/(2 \tau)$ if $|x| < \tau/2$, and $g(x) = |x|-\tau/2$ if $|x| \geq \tau$. Another nice property of the Huber error measure is that it is differentiable everywhere, unlike the $\ell_1$-norm, yet still enjoys the robustness properties as one moves away from the origin, making it less sensitive to outliers than the $\ell_2$-norm.
There are many other kinds of loss functions, known as $M$-estimators~\cite{z97}, which are widely used as loss functions in robust statistics \cite{hrrs11}.

Although several specific cases have been studied, such as entry-wise $\ell_p$ loss \cite{clmw11,swz17,cgklrw17,bkw17,bbbklw17}, weighted entry-wise $\ell_2$ loss \cite{rsw16}, and cascaded $\ell_p(\ell_2)$ loss \cite{dvtv09,cw15focs}, the landscape of general entry-wise loss functions remains elusive. There are no results known for any loss function which is not scale-invariant, much less any kind of characterization of which loss functions admit efficient algorithms. This is despite the importance of these loss functions; we refer the reader to \cite{uhzb16} for a survey of generalized low rank models. This motivates the main question in our work:

\begin{question}[General Loss Functions]\label{que:lowrank}
For a given approximation factor $\alpha > 1$, which functions $g$ allow for efficient low-rank approximation algorithms? Formally, given an $n \times d$ matrix $A$, can we find a $\rank$-$k$ matrix $B$ for which $\|A-B\|_g \leq \alpha \min_{{\rank}-k~B'} \|A-B'\|_g$, where for a matrix $C$, $\|C\|_g = \sum_{i \in [n], j \in [d]} g(C_{i,j})$? What if we also allow $B$ to have rank $\poly(k \log n)$? 
\end{question}

For Question \ref{que:lowrank}, one has $g(x) = |x|^p$ for $p$-norms, and note the Huber loss function also fits into this framework. Allowing $B$ to have slightly larger rank than $k$, namely, $\poly(k \log n)$, is often sufficient for applications as it still allows for the space savings and computational gains outlined above. These are referred to as bicriteria approximations and are the focus of our work.

\nocite{swyzz19}

{\bf Notation.} Before we present our results, let us briefly introduce the notation.
For $n\in \mathbb{Z}_{\geq 0}$, let $[n]$ denote the set $\{1,2,\cdots,n\}$.
Let $A\in\mathbb{R}^{n\times m}$.
$A_i$ and $A^j$ denote the $i^{\text{th}}$ column and the $j^{\text{th}}$ row of $A$ respectively.
Let $P\subseteq [m], Q\subseteq[n]$.
$A_P$ denotes the matrix which is composed by the columns of $A$ with column indices in $P$.
Similarly, $A^Q$ denotes the matrix composed by the rows of $A$ with row indices in $Q$.
Let $S$ be a set and $s\in\mathbb{Z}_{\geq 0}$.
We use
$S\choose s$ to denote the set of all the size-$s$ subsets of $S$.

\subsection{Our Results}
We studied low rank approximation with respect to general error measures. 
Our algorithm
is a column subset selection algorithm, returning a small subset of columns which span a good low rank approximation.
Column subset selection has the benefit of preserving sparsity and interpretability, as described above.

We give a ``zero-one law'' for such column subset selection problems.
We describe two properties on the function $g$ that we need
to obtain our low rank approximation algorithms. 
We also show that if we are missing any one of the properties, then we can find an example function $g$ for which there is no good column subset selection (see Appendix~\ref{sec:necessity}).

Since we obtain column subset selection algorithms for a wide class
of functions, our algorithms must necessarily be bicriteria and have
approximation factor at least $\poly(k)$. Indeed, a special case of our
class of functions includes entrywise $\ell_1$-low rank approximation,
for which it was shown in Theorem G.27 of \cite{swz17}
that any subset of $\poly(k)$ columns incurs
an approximation error of at least $k^{\Omega(1)}$. We also show that
for the entrywise Huber-low rank approximation, already for
$k =1$, $\sqrt{\log n}$ columns
are needed to obtain any constant factor approximation, thus showing
that for some of the functions we consider, a dependence on $n$
in our column subset size is necessary.

We note that previously for almost all such functions,
it was not known how to obtain
any non-trivial approximation factor with any sublinear number of columns.

\subsubsection{A Zero-One Law}\label{sec:zeroonelaw}

We first state three general properties, the first two of which are structural properties and are necessary and sufficient for obtaining a good approximation from a small subset of columns. The third property is needed for efficient running time. 

{\bf Approximate triangle inequality.} 
For $t\in\mathbb{Z}_{>0}$, we say a function $g(x) : \R \rightarrow \R_{\geq 0}$ satisfies the ${\ati}_{g,t}$-approximate triangle inequality if for any $x_1, x_2, \cdots, x_t \in \R$, 
$
g \left( \sum x_i \right) \leq {\ati}_{g,t} \cdot \sum g( x_i ).
$

{\bf Monotone property.}
For any parameter $\sym_g \geq 1$, we say function $g(x) : \R \rightarrow \R_{\geq 0}$ is $\sym_g$-monotone if
for any $x,y\in\mathbb{R}$ with $0\leq |x|\leq |y|,$ we have $g(x)\leq \sym_{g}\cdot g(y).$

Many functions including most $M$-estimators~\cite{z97} and the quantile function~\cite{kb78} satisfy the above two properties.
 See Table~\ref{tab:M-estimators} for several examples.
\begin{table}[t]
	\caption{ Example functions satisfying both structural properties.}
	\label{tab:M-estimators}
	\begin{center}
		\begin{small}
				\begin{tabular}{lccc}
					\toprule
					& $g(x)$ & $\ati_{g,t}$ & $\mon_g$  \\ \hline
					\small{\sc Huber} & \small{$\left\{\begin{array}{ll}x^2/2&|x|\leq \tau\\ \tau(|x|-\tau/2) & |x|>\tau\end{array}\right.$} & $O(t)$ & $1$ \\ \hline 
					 \small{\sc $\ell_p$ $(p\geq 1)$} & \small{$|x|^p/p$} &  $O(t^{p-1})$ & $1$ \\ \hline
					\small{\sc $\ell_1-\ell_2$} &  \small{$2(\sqrt{1+x^2/2}-1)$} & $O(t)$ & $1$  \\ \hline
					\small{\sc Geman-McClure} & $x^2/(2+2x^2)$ & $O(t)$ & $1$ \\ \hline
					\small{\sc ``Fair"} & \small{$\tau^2\left(|x|/\tau-\log(1+|x|/\tau)\right)$} & $O(t)$  & $1$\\ \hline
					\small{\sc Tukey} & \small{$\left\{\begin{array}{ll} \tau^2/6\cdot(1-(1-(x/\tau)^2)^3) &|x|\leq \tau\\ \tau^2/6 & |x|>\tau\end{array}\right.$} & $O(t)$ & $1$\\ \hline 
					\small{\sc Cauchy} & \small{$\tau^2/2\cdot \log(1+(x/\tau)^2)$} & $O(t)$ & $1$\\ \hline
					\small{\sc Quantile $(\tau\in(0,1))$}  & \small{$\left\{\begin{array}{ll} \tau x &x \geq 0\\ (\tau - 1)x & x<0\end{array}\right.$} & $1$ & $\max\left(\frac{\tau}{1-\tau},\frac{1-\tau}{\tau}\right)$\\ 
					\bottomrule
				\end{tabular}
		\end{small}
	\end{center}
	
\end{table}
We refer the reader to the supplementary, namely Appendix~\ref{sec:necessity}, for the necessity of these two properties. Our next property
is not structural, but rather states that if the loss function has an
efficient regression algorithm, then that suffices 
to efficiently find a small subset of columns spanning a good low
rank approximation.

{\bf Regression property.}
We say function $g(x) : \R \rightarrow \R_{\geq 0}$ has the $(\reg_{g,d},\mathcal{T}_{\reg,g,n,d,m})$-regression property if the following holds: given two matrices $A \in \R^{n \times d}$ and $B \in \R^{n \times m}$, for each $i\in [m]$, let $\OPT_i$ denote $\min_{x \in \R^{ d } } \| A x - B_i \|_g $. There is an algorithm that runs in $\mathcal{T}_{\reg,g,n,d,m}$ time and outputs a matrix $X'\in \R^{d \times m}$ such that 
$
 \| A X'_i - B_i \|_g \leq \reg_{g,d} \cdot \OPT_i, \forall i \in [m] 
$
and outputs a vector of estimated regression costs $v\in \R^d$ such that 
$
\OPT_i \leq v_i \leq \reg_{g,d} \cdot \OPT_i, \forall i \in [m].
$ 
The success probability is at least $1-1/\poly(nm)$.

Some functions for which regression itself is non-trivial are e.g.,
the $\ell_0$-loss function and Tukey function. 
The $\ell_0$-loss function corresponds to the nearest codeword
problem over the reals and has slightly better than an
$O(d)$-approximation
(\cite{bk02,apy09}, see also \cite{bkw17}).
For the Tukey function, 
\cite{cww19} shows that Tukey regression is NP-hard, and it also gives approximation algorithms. 
For discussion on regression solvers, we refer the reader to Appendix~\ref{sec:regression}.

\paragraph{Zero-one law (sufficient conditions):}
For any function, as long as the above general three properties hold, we can provide an efficient algorithm, as our following main theorem shows. 
\begin{theorem}
	\label{thm:intro_general_function_low_rank}
Given a matrix $A\in \R^{n\times n}$, let $k\geq 1,k'=2k+1$. Let $g : \R \rightarrow \R_{\geq 0}$ denote a function satisfying the  $\ati_{g,k'}$-approximate triangle inequality,  
the $\sym_g$-monotone property 
, and the $(\reg_{g,k'},\mathcal{T}_{\reg,g,n,k',n})$-regression property. 
Let $\OPT = \min_{\rank-k~A'}\| A' - A \|_g$. There is an algorithm that runs in $\wt{O}(n + \mathcal{T}_{\reg,g,n,k',n})$ time and outputs a set $S\subseteq [n]$ with $|S|= O(k \log n)$ such that
with probability at least $0.99$, 
\begin{align*}
\min_{X \in \R^{|S| \times n}} \| A_S X - A \|_g \leq \ati_{g,k'} \cdot \sym_{g} \cdot \reg_{g,k'} \cdot O(k\log k) \cdot \OPT.
\end{align*}
\end{theorem}
Although the input matrix $A$ in the above statement is a square matrix, it is straightforward to extend the result to the rectangular case.

By the above theorem, we can obtain a good subset of columns. 
To further get a low rank matrix $B$ which is a good low rank approximation to $A$, it is sufficient to take an additional $\mathcal{T}_{\reg,g,n,|S|,n}$ time to solve the regression problem.

\paragraph{Zero-one law (necessary conditions):} In Appendix~\ref{sec:no_ati}, we show how to construct a monotone function without approximate triangle inequality such that it is not possible to obtain a good low rank approximation by selecting a small subset of columns.

In Appendix~\ref{sec:no_mon}, we discuss a function which has the approximate triangle inequality but is not monotone. We show that for some matrices, there is no small subset of columns which can give a good low rank approximation for such loss function.

\subsubsection{Lower Bound on the Number of Columns}
One may wonder if the $\log n$ blowup in rank is necessary in our theorem.
We show some dependence on $n$ is necessary by showing that for the important
Huber loss function,
at least $\sqrt{\log n}$ columns are required in order to obtain a constant factor approximation for $k = 1$:
\begin{theorem}
Let $H(x)$ denote the following function:
$
H (x) =
\begin{cases}
x^2 , & \text{~if~} |x| < 1;\\
|x|, & \text{~if~} |x| \geq 1.
\end{cases}
$

For any $n\geq 1,$ 
there is a matrix $A\in \R^{n \times n}$ such that, if we select 
$o(\sqrt{\log n})$ columns to fit the entire matrix, there is 
no $O(1)$-approximation, i.e., for any subset $S\subseteq[n]$ with $|S|=o(\sqrt{\log n}),$
\begin{align*}
\min_{X \in \R^{|S| \times n} } \| A_S X -A \|_H \geq 
\omega(1)\cdot\min_{\rank-1~A'} \| A' - A \|_H.
\end{align*}
\end{theorem}
Notice that the above function $H(x)$ is always a constant approximation to the Huber function (see Table~\ref{tab:M-estimators}) with $\tau=1$.
Thus, the hardness also holds for the Huber function.
For more discussion on our lower bound, we refer the reader to Appendix~\ref{sec:hardinstance}.

\subsection{Overview of our Approach and Related Work} 
\paragraph{Low Rank Approximation for General Functions.} 
A natural approach to low rank approximation is ``column subset selection'', which has been extensively studied in numerical linear algebra \cite{dmm06,dmm06b,dmm08,bmd09,bdm11,fegk13,bw14,ws15,swz17,swz19}. One can take the 
column subset selection algorithm
for $\ell_p$-low rank approximation
in \cite{cgklrw17} and try to adapt it to general loss functions. Namely,
their argument shows that for any matrix $A \in \mathbb{R}^{n \times n}$
there exists a subset $S$ of
$k$ columns of $A$, denoted by $A_S \in \mathbb{R}^{n \times k}$,
for which there exists a $k \times n$ matrix $V$ for which
$\|A_S V - A\|_p^p \leq (k+1)^p \min_{\textrm{rank-}k~B'}\|A-B'\|_p^p$; we refer
the reader to Theorem 3 of \cite{cgklrw17}. Given the existence of
such a subset $S$, a natural next idea is to then sample a set $T$
of $k$ columns of $A$ uniformly at random.
It is then likely the case that if we look at
a random column $A_i$,
(1) with probability $1/(k+1)$,
$i$ is not among the subset $S$ of $k$ columns out of the $k+1$
columns $T \cup \{i\}$ defining the optimal rank-$k$ approximation to
the submatrix $A_{T \cup \{i\}}$, and (2) with probability at least $1/2$,
the best rank-$k$ approximation
to $A_{T \cup \{i\}}$ has cost at most
\begin{eqnarray}\label{eqn:cost}
( 2(k+1) / n ) \cdot \min_{\textrm{rank-}k~B'}\|A-B'\|_p^p.
\end{eqnarray}
Indeed, (1)
follows from $T \cup \{i\}$ being a uniformly random subset of $k+1$
columns, while (2) follows from a Markov bound. The argument in Theorem 7
of \cite{cgklrw17} is then able to ``prune'' a $1/(k+1)$ fraction of columns
(this can be optimized to a constant fraction) in expectation, by ``covering''
them with the random set $T$. Recursing on the remaining columns, this procedure
stops after $k \log n$ iterations, giving a column subset of size $O(k^2 \log n)$
(which can be optimized to $O(k \log n)$)
and an $O(k)$-approximation.

The proof in \cite{cgklrw17}
of the existence of a subset $S$ of $k$ columns of $A$ spanning
a $(k+1)$-approximation above is quite general, and one might suspect it
generalizes to a large class of error functions. Suppose, for example,
that $k = 1$. The idea there is to write $A = A^* + \Delta$, where
$A^* = U \cdot V$ is the optimal rank-$1$ $\ell_p$-low rank approximation
to $A$. One then ``normalizes'' by the error, defining
$\tilde{A}^*_i = A_i^*/\|\Delta_i\|_p$ and letting $s$ be such that
$\|\tilde{A}^*_s\|_p$ is largest. The rank-$1$ subset $S$ is
then just $A_s$. Note that since $\tilde{A}^*$ has rank-$1$
and $\|\tilde{A}^*_s\|_p$ is largest, one can write $\tilde{A}^*_j$ for
every $j \neq s$ as $\alpha_j \cdot \tilde{A}^*_s$ for $|\alpha_j| \leq 1$.
The fact that $|\alpha_j| \leq 1$ is crucial; indeed, consider what
happens when we try to ``approximate'' $A_j$ by
$A_s \cdot \frac{\alpha_j \|\Delta_j\|_p}{\|\Delta_s\|_p}$. Then
$\left\|A_j - A_s  \alpha_j \|\Delta_j\|_p / \|\Delta_s\|_p \right\|_p
\leq \|A_j - A_j^*\|_p + \left\|A_j^* - A_s \alpha_j \|\Delta_j\|_p / \|\Delta_s\|_p \right\|_p 
 = \|\Delta_j\|_p + \left\|A_j^*-(A_s^* + \Delta_s)  \alpha_j \|\Delta_j\|_p / \|\Delta_s\|_p \right\|_p 
=  \|\Delta_j\|_p + \left\| \Delta_s  \alpha_j \|\Delta_j\|_p / \|\Delta_s\|_p \right\|_p,$
and since the $p$-norm is monotonically increasing and $\alpha_j \leq 1$, the latter is at most
$\|\Delta_j\|_p + \|\Delta_s \frac{\|\Delta_j\|_p}{\|\Delta_s\|_p}\|_p$. So far, all we have used
about the $p$-norm is the monotone increasing property, so one could hope that the argument could be generalized to a much wider class of functions.

However, at this point the proof uses that the $p$-norm has {\it scale-invariance}, and so
$\|\Delta_s \frac{\|\Delta_j\|_p}{\|\Delta_s\|_p} \|_p
= \|\Delta_j\|_p \cdot \|\frac{\Delta_s}{\|\Delta_s\|_p}\|_p = \|\Delta_j\|_p$,
and it follows that $\|A_j - A_s \frac{\alpha_j \|\Delta_j\|_p}{\|\Delta_s\|_p} \|_p \leq 2 \|\Delta_j\|_p$,
giving an overall $2$-approximation (recall $k = 1$). But what would happen for a general, not necessarily
scale-invariant function $g$? We need to bound
$\|\Delta_s \frac{\|\Delta_j\|_g}{\|\Delta_s\|_g}\|_g$. If we could bound this by $O(\|\Delta_j\|_g)$,
we would obtain the same conclusion as before, up to constant factors. Consider, though,
the ``reverse Huber function'': $g(x) = x^2$ if $x \geq 1$ and $g(x) = |x|$ for $x \leq 1$. Suppose that
$\Delta_s$ and $\Delta_j$ were just $1$-dimensional vectors, i.e., real numbers, so we need to bound
$g(\Delta_s g(\Delta_j) / g(\Delta_s))$ by $O(g(\Delta_j))$. Suppose $\Delta_s = 1$. Then
$g(\Delta_s) = 1$ and $g(\Delta_s g(\Delta_j)/g(\Delta_s)) = g(g(\Delta_j))$ and if $\Delta_j = n$,
then $g(g(\Delta_j)) = n^4 = g(\Delta_j)^2$, much larger than the $O(g(\Delta_j))$ we were aiming for.

Maybe the analysis can be slightly changed to correct for these normalization issues?  
This is not the case, as we show that unlike for $\ell_p$-low rank approximation, for the reverse Huber function
{\it there is no subset of $2$ columns of $A$ obtaining better than an $n^{1/4}$-approximation factor}. (See Section~\ref{sec:reverse_huber} for more details). Further, the lack of scale invariance not only breaks the argument in \cite{cgklrw17}, it shows that combinatorially
such functions $g$ behave very differently than $\ell_p$-norms. We show more generally there exist
functions, in particular the Huber function, for which one needs to choose $\Omega(\sqrt{\log n})$ columns to obtain
a constant factor approximation; we describe this more below.
Perhaps more surprisingly, we show a subset of $O(\log n)$ columns suffice to
obtain a constant factor approximation to the best rank-$1$ approximation for any function $g(x)$ which is
approximately monotone and has the approximate triangle inequality, the latter implying for any constant $C > 0$
and any $x \in \mathbb{R}_{\geq 0}$, $g(Cx) = O(g(x))$.
For $k > 1$, these conditions become: (1) $g(x)$ is monotone non-decreasing in $x$, (2) $g(x)$ is within a $\poly(k)$
factor of $g(-x)$, and (3) for any real number $x \in \mathbb{R}^{\geq 0}$, $g(O(kx)) \leq \poly(k) \cdot g(x)$. We show
it is possible to obtain an $O(k^2 \log k)$ approximation with $O(k \log n)$ columns. We give the intuition and main lemma statements for our result in Section \ref{sec:alg}, deferring proofs to the supplementary material.

Even for $\ell_p$-low rank approximation, our algorithms slightly improve and correct a minor error in
\cite{cgklrw17} which claims in Theorem 7 an $O(k)$-approximation with $O(k \log n)$ columns for $\ell_p$-low rank approximation. However,
their algorithm actually gives an $O(k \log n)$-approximation with $O(k \log n)$ columns. In \cite{cgklrw17} it was argued that one expects
to pay a cost of $O(k/n) \cdot \min_{\textrm{rank-}k~B'}\|A-B'\|_p^p$ per column as in (\ref{eqn:cost}), and since each column is only counted
in one iteration, summing over the columns gives $O(k) \cdot \min_{\textrm{rank-}k~B'}\|A-B'\|_p$ total cost. The issue is that
the value of $n$ is changing in each iteration, so if in the $i$-th iteration it is $n_i$, then we could pay
$n_i \cdot O(k/n_i) \cdot \min_{\textrm{rank-}k~B'}\|A-B'\|_p = O(k) \cdot  \min_{\textrm{rank-}k~B'}\|A-B'\|_p$ in each of $O(\log n)$ iterations,
giving $O(k \log n)$ approximation ratio. In contrast, our algorithm achieves an $O(k \log k)$ approximation ratio for
$\ell_p$-low rank approximation as a special case, which
gives the first $O(1)$ approximation in nearly linear time for any constant $k$ for $\ell_p$ norms. Our analysis is finer in that we show
not only do we expect to pay a cost of $O(k/n_i) \cdot \min_{\textrm{rank-}k~B'}\|A-B'\|_p^p$ per column in iteration $i$, we pay
$O(k/n_i)$ times the cost of the best rank-$k$ approximation to $A$ {\it after the most costly} $n/k$ columns have been removed; thus
we pay $O(k/n_i)$ times a {\it residual cost} with the top $n/k$ columns removed. This ultimately implies any column's cost can contribute
in at most $O(\log k)$ of $O(\log n)$ recursive calls, replacing an $O(\log n)$ factor with an $O(\log k)$ factor in the approximation ratio. This also gives the first $\poly(k)$-approximation for $\ell_0$-low rank approximation, studied
in \cite{bkw17}, improving the $O(k^2 \log (n/k))$-approximation there to $O(k^2 \log k)$ and giving the first constant
approximation for constant $k$.

\section{Algorithm for General Loss Low Rank Approximation}\label{sec:alg}

Our algorithm is presented in Algorithm~\ref{alg:general_function_low_rank}.
First, let us briefly analyze the running time.
Consider fixed $i\in[r],j\in[\log n]$.
Sampling $S_i^{(j)}$ takes $O(k)$ time.
Solving $\reg_{g,2k}$-approximate regression $\min_{x}\|A_{S_i^{(j)}}x-A_t\|_g$ for all $t\in T_{i-1}\setminus S_{i}^{(j)}$ takes $\mathcal{T}_{\reg,g,n,2k,|T_{i-1}\setminus S_i^{(j)}|}\leq \mathcal{T}_{\reg,g,n,2k+1,n}$ time.
Since finding $|T_{i-1}\setminus S_i^{(j)}|/20$ smallest element can be done in $O(n)$ time, $R_i^{(j)}$ can be computed in $O(n)$ time.
Thus the inner loop takes $O(n+\mathcal{T}_{\reg,g,n,2k+1,n})$ time.
Since $r=O(\log n)$, the total running time over all $i,j$ is $O((n+\mathcal{T}_{\reg,g,n,2k+1,n})\log^2 n)$. 
In the remainder of the section, we will sketch the proof of the correctness. 
For the missing proofs, we refer the reader to Appendix~\ref{sec:missing_proofs}.

\begin{algorithm}[t]
	\begin{algorithmic}[1]\caption{Low rank approximation algorithm for general functions}\label{alg:general_function_low_rank}
		\Procedure{\textsc{GeneralFunctionLowRankApprox}}{$A\in\mathbb{R}^{n\times n},k\in\mathbb{Z}_{\geq 1},g:\mathbb{R}\rightarrow \mathbb{R}_{\geq 0}$} 
		\State Initialization: $T_0 \leftarrow [n],i\gets 1, r\gets 0$
		\For{$|T_{i-1}|\geq 1000 k$}
		\For{$j = 1 \to \log n$}
		\State Sample $S^{(j)}_i$ from ${ T_{i-1} \choose 2k}$ uniformly at random \label{sta:sampling}
		\State Solve the $\reg_{g,2k}$-approximate regression  $\min_{x\in\mathbb{R}^{2k}}\|A_{S_i^{(j)}}x-A_t\|_g$ for each $t\in T_{i-1}\setminus S_i^{(j)}$, and let $v^{(j)}_{i,t}$ be the $\reg_{g,2k}$-estimated regression cost
		\Comment{See Section~\ref{sec:zeroonelaw} for regression property}
		\State $R_i^{(j)} \leftarrow
		\{t\mid v_{i,t}^{(j)}\text{ is the bottom }|T_{i-1}\setminus S_i^{(j)}|/20\text{ largest value in }\{v_{i,t'}^{(j)}\mid t'\in T_{i-1}\setminus S_i^{(j)}\} \}$ 
		\State $c_{i}^{(j)} \leftarrow \sum_{t \in R_i^{(j)}} v_{i,t}^{(j)} $
		\EndFor
		\State $j^* \leftarrow \arg\min_{j \in [\log n]} \left\{ c_i^{(j)} \right\}$
		\State $S_i \leftarrow S^{(j^*)}_i,R_i\leftarrow R^{(j^*)}_i,T_i \leftarrow T_{i-1} \backslash \left(S_i \cup R_i\right)$
		\State $r\gets i$
		\State  $i\gets i + 1$
		\EndFor
		\State \Return $S=T_r\cup\bigcup_{i\in[r]} S_i$ \Comment{It is easy to see $r\leq O(\log n)$ from the above procedure}
		\EndProcedure
	\end{algorithmic}
\end{algorithm}

\subsection{Properties of Uniform Column Sampling}
Let us first introduce some useful notation.
Consider a rank-$k$ matrix $M^* \in \R^{n \times m}$. 
For a set $H \subseteq [m]$, let $R_{M^*}(H) \subseteq H$ be a set such that
\begin{align*}
R_{M^*}(H) = \arg\max_{P:P \subseteq H} \left\{ \left|\det\left( (M^*)_P^Q \right)\right| ~\bigg|~ |P| = |Q|= \mathrm{rank} (M^*_H), Q \subseteq [n]  \right\}.
\end{align*}
where $\det(C)$ denotes the determinant of a square matrix $C$.
Notice that in the above formula, the maximum is over all possible choices of $P$ and $Q$ while $R_{M^*}(H)$ only takes the value of the corresponding $P$.
By Cramer's rule, if we use a linear combination of the columns of $M^*_{R_{M^*}(H)}$ to express any column of $M^*_H$, the absolute value of every fitting coefficient will be at most $1$.
For example, consider a rank $k$ matrix $M^*\in\mathbb{R}^{n\times(k+1)}$ and $H=[k+1]$.
Let $P\subseteq [k+1],Q\subseteq [n],|P|=|Q|=k$ be such that $|\det({(M^*)}_P^Q)|$ is maximized.
Since $M^*$ has rank $k$, we know $\det({(M^*)}_P^Q)\not= 0$ and thus the columns of $M^*_P$ are independent.
Let $i\in [k+1]\setminus P$.
Then the linear equation $M^*_Px=M^*_i$ is feasible and
there is a unique solution $x$.
Furthermore, by Cramer's rule $x_j=\frac{\det({(M^*)}^Q_{[k+1]\setminus\{j\}})}{\det({(M^*)}_P^Q)}$. 
Since $|\det({(M^*)}_P^Q)|\geq |\det({(M^*)}^Q_{[k+1]\setminus\{j\}})|$, we have $\|x\|_{\infty}\leq 1$.

Consider an arbitrary matrix $M\in\mathbb{R}^{n\times m}$. 
We can write $M=M^*+N,$ where $M^*\in\mathbb{R}^{n\times m}$ is an arbitrary rank-$k$ matrix, and $N\in\mathbb{R}^{n\times m}$ is the residual matrix.
The following lemma shows that, if we randomly choose a subset $H\subseteq[m]$ of $2k$ columns, and we randomly look at another column $i,$ then with constant probability, the absolute values of all the coefficients of using a linear combination of the columns of $M^*_H$ to express $M^*_i$ are at most $1$, and furthermore, if we use the same coefficients to use columns of $M_H$ to fit $M_i$, then the fitting cost is proportional to $\|N_H\|_g+\|N_i\|_g.$
 \begin{lemma}\label{lem:pro_random_S_i_two_parts}
 Given a matrix $M \in \R^{n \times m}$ and a parameter $k\geq 1$, 
 let $M^*\in\mathbb{R}^{n\times m}$ be an arbitrary rank-$k$ matrix.
Let $N = M - M^*$. Let $H \subseteq [m]$ be a uniformly random subset of $[m]$, and let $i$ denote a uniformly random index sampled from $[m] \backslash H$. 
Then 
$
\mathrm{(\RN{1})}  \Pr
\left[ i \notin R_{M^*}(H\cup \{i\}) \right] \geq 1/2;
$
$\mathrm{(\RN{2})}$ If $i\notin R_{M^*}(H \cup \{i\})$, then there exist $|H|$ coefficients $\alpha_1, \alpha_2, \cdots, \alpha_{|H|}$ for which
$
M_i^* = \sum_{j=1}^{|H| } \alpha_j (M_H^*)_j, \forall j\in [|H|], |\alpha_j|\leq 1,
$
and
$
\min_{x\in \R^{|H|}} \| M_H x - M_i\|_g \leq \ati_{g,|H|+1} \cdot \sym_g \cdot \left( \| N_i \|_g + \sum_{j=1}^{|H|} \| (N_H)_j \|_g \right).
$
\end{lemma}

Notice that part $\mathrm{(\RN{2})}$ of the above lemma does not depend on any randomness of $H$ or $i$. 
By applying part $\mathrm{(\RN{1})}$ of the above lemma, it is enough to prove that if we randomly choose a subset $H$ of $2k$ columns, there is a constant fraction of columns that each column $M_i^*$ can be expressed by a linear combination of columns in $M_H^*,$ and the absolute values of all the fitting coefficients are at most $1$. 
Because of Cramer's rule, it thus suffices to prove the following lemma.
\begin{lemma}\label{lem:pro_i_not_in_R_S_i}
\begin{align*}
\Pr_{H \sim {{[m]}\choose{2k}}} \biggl[ \biggl| \biggl\{ i ~\bigg|~ i\in [m]\setminus H,i\not\in R_{M^*}(H\cup \{i\}) \biggr\} \biggr|\geq (m-2k)/4\biggr] \geq 1/4.
\end{align*}
\end{lemma}

\subsection{Correctness of the Algorithm}

We write the input matrix $A$ as $A^*+\Delta$, where $A^*\in\mathbb{R}^{n\times n}$ is the best rank-$k$ approximation to $A$, and $\Delta\in\mathbb{R}^{n\times n}$ is the residual matrix with respect to $A^*$. 
Then $\|\Delta\|_g=\sum_{i=1}^n \|\Delta_i\|_g$ is the optimal cost.
As shown in Algorithm~\ref{alg:general_function_low_rank}, our approach iteratively eliminates all the columns. 
In each iteration, we sample a subset of columns, and use these columns to fit other columns.
We drop a constant fraction of columns which have a good fitting cost.
Suppose the indices of the columns surviving after the $i$-th outer iteration are $T_i=\{ t_{i,1}, t_{i,2}, \cdots, t_{i,m_i} \} \subseteq [n].$ 
Without loss of generality, we can assume $\|\Delta_{t_{i,1}}\|_g\geq \|\Delta_{t_{i,2}}\|_g\geq \cdots\geq \|\Delta_{t_{i,m_i}}\|_g$.
The following claim shows that if we randomly sample $2k$ column indices $H$ from $T_i,$ then the cost of $\Delta_{H}$ will not be large.

\begin{claim}\label{cla:t_i_greater_than_19_over_20}
If $|T_i|=m_i \geq 1000k$,
$\Pr_{H\sim {T_i \choose 2k}} \left[ \sum_{j\in H} \| \Delta_j \|_g \leq 400 \frac{k}{m_i} \sum_{j= \frac{m_i}{100k} }^{m_i} \| \Delta_{t_{i,j} }\|_g \right] \geq \frac{19}{20}.$
\end{claim}

By an averaging argument, in the following claim, we can show that there is a constant fraction of columns in $T_i$ whose optimal cost is also small.

\begin{claim}\label{cla:t_i_less_than_m_over_5}
If $|T_i|=m_i \geq 1000k$,  
$\left| \left\{ t_{i,j} ~\bigg|~ t_{i,j} \in T_i, \| \Delta_{t_{i,j}} \|_g \geq \frac{20}{m_i} \sum_{j'=\frac{m_i}{100k} }^{m_i} \| \Delta_{t_{i,j'}} \|_g \right\} \right| \leq \frac{1}{5} m_i.$

\end{claim}

By combining Lemma~\ref{lem:pro_i_not_in_R_S_i}, part (\RN{2}) of Lemma~\ref{lem:pro_random_S_i_two_parts} with the above two claims, it is sufficient to prove the following core lemma.
It says that if we randomly choose a subset of $2k$ columns from $T_i$, then we can fit a constant fraction of the columns from $T_i$ with a small cost.

\begin{lemma}\label{lem:low_fitting_cost} 
If $|T_i|=m_i \geq 1000 k$, 
\begin{align*}
\Pr_{H \sim {T_i \choose 2k}} \left[ \left| \left\{ j ~\bigg|~ j\in T_i, \min_{x\in \R^{|H|}}\| A_H x - A_j \|_g \leq C_1\cdot \frac{1}{m_i} \cdot \sum_{j'=\frac{m_i}{100k} }^{m_i} \| \Delta_{t_{i,j'}} \|_g \right\} \right|  \geq \frac{1}{20} m_i \right] \geq \frac{1}{5},
\end{align*}
where
$
C_1 = 500 \cdot k \cdot \ati_{g,|S|+1} \cdot \sym_g .
$
\end{lemma}

Let us briefly explain why the above lemma is enough to prove the correctness of our algorithm.
For each column $j\in[m]$, either the column $j$ is in $T_r$ and is selected by the end of the algorithm, or $\exists i<r$ such that $j\in T_i\setminus T_{i+1}$.
If $j\in T_i\setminus T_{i+1}$, then by the above lemma, we can show that with high probability, $\min_x\|A_{S_{i+1}}x-A_j\|_g\leq O(C_1 \|\Delta\|_1/|T_i|)$. 
Thus, $\min_X\|A_{S_{i+1}}X-A_{T_i\setminus T_{i+1}}\|_g\leq O(C_1 \|\Delta\|_1)$.
It directly gives a $O(rC_1)=O(C_1\log n)$ approximation. 
For the detailed proof of Theorem~\ref{thm:intro_general_function_low_rank}, we refer the reader to Appendix~\ref{sec:missing_proofs}.

\section{Experiments}\label{sec:exp}

\begin{figure*}[!t]
	\centering
	\begin{tabular}{cccc}
		\includegraphics[width=0.23\textwidth]{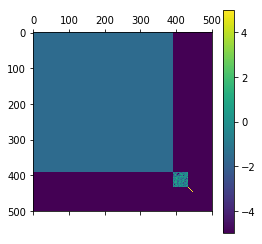}&
		\includegraphics[width=0.23\textwidth]{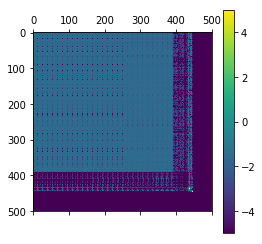} &
		 \includegraphics[width=0.23\textwidth]{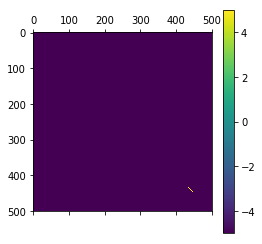}&
		\includegraphics[width=0.23\textwidth]{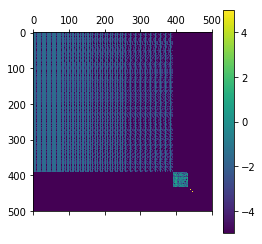}\\
		(a) & (b) & (c) & (d)\\
	\end{tabular}
	
	\caption{ The input data has size $500\times 500$. The color indicates the logarithmic magnitude of the absolute value of each entry. (a) is the input matrix. It contains $3$ blocks on its diagonal. The top-left one has uniformly small noise. The central one is the ground truth. The bottom-right one contains sparse outliers. Each block has rank $14$. So the rank of the input matrix is $3\times 14 = 42.$ (b) is the entry-wise $\ell_1$ loss rank-$14$ approximation given by~\cite{swz17}. As shown above, it mainly covers the small noise, but loses the information of the ground truth. (c) is the Frobenius norm rank-$14$ approximation given by the top $14$ singular vectors. As shown in the figure, it mainly covers the outliers. However, it loses the information of the ground truth. (d) is the rank-$1$ bi-criteria solution given by our algorithm. As we can see, it can cover the ground truth matrix quite well. }
	\label{fig:fig1} 
\end{figure*}

\begin{figure*}[!ht]
	\centering
	\begin{tabular}{cc}
		\includegraphics[width=0.45\textwidth]{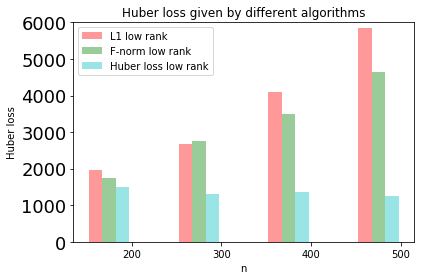}
	\end{tabular}
	
	\caption{ The Huber loss given by different algorithms. The red bar is for the entrywise $\ell_1$ low rank approximation algorithm~\cite{swz17}. The green bar is for traditional PCA. The blue bar is for our algorithm. For input size $n=200,300$, all the algorithms output rank-$12$ approximations. For input size $n= 400,500$, all the algorithms output rank-$14$ approximations. }
	\label{fig:fig2} 
\end{figure*}

We show that with the Huber loss low rank approximation, it is possible to outperform the SVD and entrywise $\ell_1$-low rank approximation on certain noise distributions. Even bi-criteria solutions can work very well. This motivates our study of general entry-wise loss functions. 

Suppose the noise of the input matrix is a mixture of small Gaussian noise and sparse outliers. Consider an extreme case: the data matrix $A\in\mathbb{R}^{n\times n}$ is a block diagonal matrix which contains three blocks: one block has size $n_1\times n_1$ $(n_1=\Theta(n))$ which has uniformly small noise (every entry is $\Theta(1/\sqrt{n})$), another block has only one entry which is a large outlier (with value $\Theta(n^{0.8})$), and the third matrix is the ground truth matrix with size $n_3\times n_3$ $(n_3 = \Theta(n^{0.6}))$ where the absolute value of each entry is at least $1/n^{o(1)}$ and at most $n^{o(1)}.$ If we apply Frobenius norm rank-$1$ approximation, then since $(n^{0.8})^2>(n^{0.6})^2\cdot n^{o(1)}$ and $(n^{0.8})^2> n^2\cdot (1/\sqrt{n})^2$, we can only learn the large outlier. If we apply entry-wise $\ell_1$ norm rank-$1$ approximation, then since $n^2\cdot 1/\sqrt{n} > (n^{0.6})^2\cdot n^{o(1)}$ and $n^2\cdot 1/\sqrt{n} > n^{0.8}$, we can only learn the uniformly small noise. But if we apply Huber loss rank-$1$ approximation, then we can learn the ground truth matrix.

A natural question is: can bi-criteria Huber loss low rank approximation also learn the ground truth matrix under certain noise distributions? We did experiments to answer this question.

{\bf Parameters.} In each iteration, we choose $2k$ columns to fit the remaining columns, and we drop half of the columns with smallest regression cost. In each iteration, we repeat $20$ times to find the best $2k$ columns. At the end, if there are at most $4k$ columns remaining, we finish our algorithm. We choose to optimize the Huber loss function, i.e., $f(x) = \frac{1}{2} x^2$ for $x\leq 1,$ and $f(x) = |x| - \frac{1}{2}$ for $x>1.$

{\bf Data.} We evaluate our algorithms on several input data matrix $A\in\mathbb{R}^{n\times n}$ sizes, for $n\in\{200,300,400,500\}.$ For rank-$1$ bi-criteria solutions, the output rank is given in Table~\ref{tab:output_rank}.

\begin{table}[h!]
	\center
	\caption{The output rank of our algorithm for different input sizes and for $k=1$. }\label{tab:output_rank}
	\begin{tabular}{ c | c | c | c | c }
		\hline
	$n$	& 200 & 300 & 400 & 500 \\
		\hline
	Output $\rank$  & 12  & 12  & 14  & 14 \\
		\hline
	\end{tabular}
\end{table}
$A$ is constructed as a block diagonal matrix with three blocks. The first block has size $\frac{4}{5}n\times \frac{4}{5}n$. It contains many copies of $k'$ different columns where $k'$ is equal to the output rank corresponding to $n$ (see Table~\ref{tab:output_rank}). The entry of a column is uniformly drawn from $\{-5/\sqrt{n},5/\sqrt{n}\}$. The second block is the ground truth matrix. It is generated by $1/\sqrt{k'}\cdot U\cdot V^{\top}$ where $U,V\in\mathbb{R}^{n\times k'}$ are two i.i.d. random Gaussian matrices. The last block is a size $k'\times k'$ diagonal matrix where each diagonal entry is a sparse outlier with magnitude of absolute value $5\cdot n^{0.8}$.

{\bf Experimental Results.} We compare our algorithm with Frobenius norm low rank approximation and entry-wise $\ell_1$ loss low rank approximation algorithms~\cite{swz17}. To make it comparable, we set the target rank of previous algorithms to be the output rank of our algorithm. In Figure~\ref{fig:fig1}, we can see that the ground truth matrix is well covered by our Huber loss low rank approximation. In Figure~\ref{fig:fig2}, we show that our algorithm indeed gives a good solution with respect to the Huber loss.

\paragraph{Acknowledgments.}
David P. Woodruff was supported in part by Office of Naval Research (ONR) grant N00014- 18-1-2562. Part of this work was done while he was visiting the Simons Institute for the Theory of Computing.
Peilin Zhong was supported in part by NSF grants (CCF-1703925, CCF-1421161, CCF-1714818, CCF-1617955 and CCF-1740833), Simons Foundation (\#491119 to Alexandr Andoni), Google Research Award and a Google Ph.D. fellowship.
Part of this work was done 
while Zhao Song and Peilin Zhong were interns at IBM Research - Almaden and 
while Zhao Song was visiting the Simons Institute for the Theory of Computing.

\newpage

\ifdefined\isarxivversion
\bibliographystyle{alpha}
\bibliography{ref}
\else
\bibliographystyle{unsrt}
\bibliography{ref}

\newcommand{\etalchar}[1]{$^{#1}$}
\begin{thebibliography}{CLMW11}

\bibitem[APY09]{apy09}
Noga Alon, Rina Panigrahy, and Sergey Yekhanin.
\newblock Deterministic approximation algorithms for the nearest codeword
  problem.
\newblock In {\em Algebraic Methods in Computational Complexity}, 2009.

\bibitem[BBB{\etalchar{+}}19]{bbbklw17}
Frank Ban, Vijay Bhattiprolu, Karl Bringmann, Pavel Kolev, Euiwoong Lee, and
  David~P. Woodruff.
\newblock A {PTAS} for $\ell_p$-low rank approximation.
\newblock In {\em SODA}, 2019.

\bibitem[BDM11]{bdm11}
Christos Boutsidis, Petros Drineas, and Malik Magdon{-}Ismail.
\newblock Near optimal column-based matrix reconstruction.
\newblock In {\em {IEEE} 52nd Annual Symposium on Foundations of Computer
  Science (FOCS), 2011, Palm Springs, CA, USA, October 22-25, 2011}, pages
  305--314. \url{https://arxiv.org/pdf/1103.0995}, 2011.

\bibitem[BDN15]{bdn15}
Jean Bourgain, Sjoerd Dirksen, and Jelani Nelson.
\newblock Toward a unified theory of sparse dimensionality reduction in
  euclidean space.
\newblock In {\em Proceedings of the Forty-Seventh Annual {ACM} on Symposium on
  Theory of Computing, {STOC} 2015, Portland, OR, USA, June 14-17, 2015}, pages
  499--508, 2015.

\bibitem[BHL18]{bhl18}
Peter Bartlett, Dave Helmbold, and Phil Long.
\newblock Gradient descent with identity initialization efficiently learns
  positive definite linear transformations.
\newblock In {\em International Conference on Machine Learning}, pages
  520--529, 2018.

\bibitem[BK02]{bk02}
Piotr Berman and Marek Karpinski.
\newblock Approximating minimum unsatisfiability of linear equations.
\newblock In {\em Proceedings of the Thirteenth Annual {ACM-SIAM} Symposium on
  Discrete Algorithms, January 6-8, 2002, San Francisco, CA, {USA.}}, pages
  514--516, 2002.

\bibitem[BKW17]{bkw17}
Karl Bringmann, Pavel Kolev, and David~P. Woodruff.
\newblock Approximation algorithms for ${\ell}_{0}$-low rank approximation.
\newblock In {\em Advances in Neural Information Processing Systems (NIPS)},
  pages 6651--6662, 2017.

\bibitem[BMD09]{bmd09}
Christos Boutsidis, Michael~W Mahoney, and Petros Drineas.
\newblock An improved approximation algorithm for the column subset selection
  problem.
\newblock In {\em Proceedings of the twentieth Annual ACM-SIAM Symposium on
  Discrete Algorithms (SODA)}, pages 968--977. Society for Industrial and
  Applied Mathematics, \url{https://arxiv.org/pdf/0812.4293}, 2009.

\bibitem[BW14]{bw14}
Christos Boutsidis and David~P Woodruff.
\newblock Optimal cur matrix decompositions.
\newblock In {\em Proceedings of the 46th Annual ACM Symposium on Theory of
  Computing (STOC)}, pages 353--362. ACM,
  \url{https://arxiv.org/pdf/1405.7910}, 2014.

\bibitem[CGK{\etalchar{+}}17]{cgklrw17}
Flavio Chierichetti, Sreenivas Gollapudi, Ravi Kumar, Silvio Lattanzi, Rina
  Panigrahy, and David~P Woodruff.
\newblock Algorithms for $\ell_p$ low rank approximation.
\newblock In {\em ICML}. arXiv preprint arXiv:1705.06730, 2017.

\bibitem[CLMW11]{clmw11}
Emmanuel~J Cand{\`e}s, Xiaodong Li, Yi~Ma, and John Wright.
\newblock Robust principal component analysis?
\newblock {\em Journal of the ACM (JACM)}, 58(3):11, 2011.

\bibitem[Coh16]{c16}
Michael~B. Cohen.
\newblock Nearly tight oblivious subspace embeddings by trace inequalities.
\newblock In {\em Proceedings of the Twenty-Seventh Annual {ACM-SIAM} Symposium
  on Discrete Algorithms (SODA), Arlington, VA, USA, January 10-12, 2016},
  pages 278--287, 2016.

\bibitem[CW87]{cw87}
Don Coppersmith and Shmuel Winograd.
\newblock Matrix multiplication via arithmetic progressions.
\newblock In {\em Proceedings of the nineteenth annual ACM symposium on Theory
  of computing}, pages 1--6. ACM, 1987.

\bibitem[CW13]{cw13}
Kenneth~L. Clarkson and David~P. Woodruff.
\newblock Low rank approximation and regression in input sparsity time.
\newblock In {\em Symposium on Theory of Computing Conference, STOC'13, Palo
  Alto, CA, USA, June 1-4, 2013}, pages 81--90.
  \url{https://arxiv.org/pdf/1207.6365}, 2013.

\bibitem[CW15a]{cw15focs}
Kenneth~L Clarkson and David~P Woodruff.
\newblock Input sparsity and hardness for robust subspace approximation.
\newblock In {\em 2015 IEEE 56th Annual Symposium on Foundations of Computer
  Science (FOCS)}, pages 310--329. IEEE,
  \url{https://arxiv.org/pdf/1510.06073}, 2015.

\bibitem[CW15b]{cw15soda}
Kenneth~L Clarkson and David~P Woodruff.
\newblock Sketching for m-estimators: A unified approach to robust regression.
\newblock In {\em Proceedings of the Twenty-Sixth Annual ACM-SIAM Symposium on
  Discrete Algorithms (SODA)}, pages 921--939. SIAM, 2015.

\bibitem[CWW19]{cww19}
Kenneth~L. Clarkson, Ruosong Wang, and David~P. Woodruff.
\newblock Dimensionality reduction for tukey regression.
\newblock In {\em ICML}, 2019.

\bibitem[DMM06a]{dmm06b}
Petros Drineas, Michael~W. Mahoney, and S.~Muthukrishnan.
\newblock Subspace sampling and relative-error matrix approximation:
  Column-based methods.
\newblock In {\em Approximation, Randomization, and Combinatorial Optimization.
  Algorithms and Techniques, 9th International Workshop on Approximation
  Algorithms for Combinatorial Optimization Problems, {APPROX} 2006 and 10th
  International Workshop on Randomization and Computation, {RANDOM} 2006,
  Barcelona, Spain, August 28-30 2006, Proceedings}, pages 316--326, 2006.

\bibitem[DMM06b]{dmm06}
Petros Drineas, Michael~W. Mahoney, and S.~Muthukrishnan.
\newblock Subspace sampling and relative-error matrix approximation:
  Column-row-based methods.
\newblock In {\em Algorithms - {ESA} 2006, 14th Annual European Symposium,
  Zurich, Switzerland, September 11-13, 2006, Proceedings}, pages 304--314,
  2006.

\bibitem[DMM08]{dmm08}
Petros Drineas, Michael~W. Mahoney, and S.~Muthukrishnan.
\newblock Relative-error {CUR} matrix decompositions.
\newblock {\em {SIAM} J. Matrix Analysis Applications}, 30(2):844--881, 2008.

\bibitem[DVTV09]{dvtv09}
Amit Deshpande, Kasturi~R. Varadarajan, Madhur Tulsiani, and Nisheeth~K.
  Vishnoi.
\newblock Algorithms and hardness for subspace approximation.
\newblock {\em CoRR}, abs/0912.1403, 2009.

\bibitem[FEGK13]{fegk13}
Ahmed~K Farahat, Ahmed Elgohary, Ali Ghodsi, and Mohamed~S Kamel.
\newblock Distributed column subset selection on mapreduce.
\newblock In {\em 2013 IEEE 13th International Conference on Data Mining
  (ICDM)}, pages 171--180. IEEE, 2013.

\bibitem[GKM18]{gkm18}
Surbhi Goel, Adam Klivans, and Raghu Meka.
\newblock Learning one convolutional layer with overlapping patches.
\newblock In {\em ICML}. arXiv preprint arXiv:1802.02547, 2018.

\bibitem[HRRS11]{hrrs11}
Frank~R Hampel, Elvezio~M Ronchetti, Peter~J Rousseeuw, and Werner~A Stahel.
\newblock {\em Robust statistics: the approach based on influence functions},
  volume 196.
\newblock John Wiley \& Sons, 2011.

\bibitem[Hub64]{h64}
Peter~J. Huber.
\newblock Robust estimation of a location parameter.
\newblock {\em The Annals of Mathematical Statistics}, 35(1):73--101, 1964.

\bibitem[KBJ78]{kb78}
Roger Koenker and Gilbert Bassett~Jr.
\newblock Regression quantiles.
\newblock {\em Econometrica: journal of the Econometric Society}, pages 33--50,
  1978.

\bibitem[LG14]{lg14}
Fran{\c{c}}ois Le~Gall.
\newblock Powers of tensors and fast matrix multiplication.
\newblock In {\em Proceedings of the 39th international symposium on symbolic
  and algebraic computation}, pages 296--303. ACM, 2014.

\bibitem[MM13]{mm13}
Xiangrui Meng and Michael~W Mahoney.
\newblock Low-distortion subspace embeddings in input-sparsity time and
  applications to robust linear regression.
\newblock In {\em Proceedings of the forty-fifth annual ACM symposium on Theory
  of computing}, pages 91--100. ACM, \url{https://arxiv.org/pdf/1210.3135},
  2013.

\bibitem[NN13]{nn13}
Jelani Nelson and Huy~L Nguy{\^e}n.
\newblock Osnap: Faster numerical linear algebra algorithms via sparser
  subspace embeddings.
\newblock In {\em 2013 IEEE 54th Annual Symposium on Foundations of Computer
  Science (FOCS)}, pages 117--126. IEEE, \url{https://arxiv.org/pdf/1211.1002},
  2013.

\bibitem[RSW16]{rsw16}
Ilya Razenshteyn, Zhao Song, and David~P Woodruff.
\newblock Weighted low rank approximations with provable guarantees.
\newblock In {\em Proceedings of the 48th Annual Symposium on the Theory of
  Computing (STOC)}, 2016.

\bibitem[Sar06]{s06}
Tam{\'{a}}s Sarl{\'{o}}s.
\newblock Improved approximation algorithms for large matrices via random
  projections.
\newblock In {\em 47th Annual {IEEE} Symposium on Foundations of Computer
  Science (FOCS) , 21-24 October 2006, Berkeley, California, USA, Proceedings},
  pages 143--152, 2006.

\bibitem[Str69]{s69}
Volker Strassen.
\newblock Gaussian elimination is not optimal.
\newblock {\em Numerische Mathematik}, 13(4):354--356, 1969.

\bibitem[SWY{\etalchar{+}}19]{swyzz19}
Zhao Song, Ruosong Wang, Lin~F Yang, Hongyang Zhang, and Peilin Zhong.
\newblock Efficient symmetric norm regression via linear sketching.
\newblock {\em arXiv preprint arXiv:1910.01788}, 2019.

\bibitem[SWZ17]{swz17}
Zhao Song, David~P Woodruff, and Peilin Zhong.
\newblock Low rank approximation with entrywise $\ell_1$-norm error.
\newblock In {\em Proceedings of the 49th Annual Symposium on the Theory of
  Computing (STOC)}. ACM, \url{https://arxiv.org/pdf/1611.00898}, 2017.

\bibitem[SWZ19]{swz19}
Zhao Song, David~P Woodruff, and Peilin Zhong.
\newblock Relative error tensor low rank approximation.
\newblock In {\em SODA}. \url{https://arxiv.org/pdf/1704.08246}, 2019.

\bibitem[UHZ{\etalchar{+}}16]{uhzb16}
Madeleine Udell, Corinne Horn, Reza Zadeh, Stephen Boyd, et~al.
\newblock Generalized low rank models.
\newblock {\em Foundations and Trends{\textregistered} in Machine Learning},
  9(1):1--118, 2016.

\bibitem[Wil12]{w12}
Virginia~Vassilevska Williams.
\newblock Multiplying matrices faster than coppersmith-winograd.
\newblock In {\em Proceedings of the forty-fourth annual ACM symposium on
  Theory of computing (STOC)}, pages 887--898. ACM, 2012.

\bibitem[WS15]{ws15}
Yining Wang and Aarti Singh.
\newblock Column subset selection with missing data via active sampling.
\newblock In {\em The 18th International Conference on Artificial Intelligence
  and Statistics (AISTATS)}, pages 1033--1041, 2015.

\bibitem[WZ13]{wz13}
David~P. Woodruff and Qin Zhang.
\newblock Subspace embeddings and $\ell_p$-regression using exponential random
  variables.
\newblock In {\em {COLT} 2013 - The 26th Annual Conference on Learning Theory,
  June 12-14, 2013, Princeton University, NJ, {USA}}, pages 546--567, 2013.

\bibitem[Yel14]{y14}
Yelp.
\newblock Yelp dataset.
\newblock {\em http://www.yelp.com/dataset\_challenge}, 2014.

\bibitem[YMM14]{ymm14}
Jiyan Yang, Xiangrui Meng, and Michael~W. Mahoney.
\newblock Quantile regression for large-scale applications.
\newblock {\em {SIAM} J. Scientific Computing}, 36(5), 2014.

\bibitem[Zha97]{z97}
Zhengyou Zhang.
\newblock Parameter estimation techniques: A tutorial with application to conic
  fitting.
\newblock {\em Image and vision Computing}, 15(1):59--76, 1997.

\bibitem[ZMKL05]{zmkl05}
Cai-Nicolas Ziegler, Sean~M McNee, Joseph~A Konstan, and Georg Lausen.
\newblock Improving recommendation lists through topic diversification.
\newblock In {\em Proceedings of the 14th international conference on World
  Wide Web}, pages 22--32. ACM, 2005.

\bibitem[ZSJ{\etalchar{+}}17]{zsjbd17}
Kai Zhong, Zhao Song, Prateek Jain, Peter~L Bartlett, and Inderjit~S Dhillon.
\newblock Recovery guarantees for one-hidden-layer neural networks.
\newblock In {\em ICML}. \url{https://arxiv.org/pdf/1706.03175.pdf}, 2017.

\end{thebibliography}
\fi

\newpage
\appendix
\section{Missing Proofs in Section~\ref{sec:alg}}\label{sec:missing_proofs}

\subsection{Proof of Lemma~\ref{lem:pro_random_S_i_two_parts}}

\begin{proof}
	(\RN{1}) Since $\rank(M^*_{H \cup \{i\}}) \leq \rank(M^*) = k$,  $|R_{M^*}(H \cup \{i\})|\leq k$. Note that $H$ is sampled from ${[m]\choose {2k}}$ uniformly at random, $i$ is sampled from $[m] \backslash H$ uniformly at random, and $|H\cup \{i\}| = 2k+1$. 
	By symmetry we have
	\begin{align*}
	\Pr_{H\sim{[m]\choose {2k}}, i \sim [m] \backslash H } \left[ i \notin R_{M^*}(H \cup \{ i \}) \right] \geq 1- \frac{k}{2k+1} \geq 1/2.
	\end{align*}
	
	(\RN{2})
	Since $i\not\in R_{M^*}(H\cup \{i\})$, by Cramer's rule, there exist $\alpha_1,\alpha_2,\cdots,\alpha_{|H|}$ such that $M^*_i=\sum_{j=1}^{|H|} \alpha_j\cdot (M^*_H)_j$ and $\forall j\in[|H|],|\alpha_j|\leq 1$.
	Then we have
	\begin{align*}
	\min_{x\in \R^{|H|}} \| M_H x - M_i \|_g \leq & ~ \left\| \sum_{j=1}^{|H|} (M_H)_j \alpha_j - M_i \right\|_g \\
	= & ~ \left\|  \sum_{j=1}^{|H|} (M_H^*)_j \alpha_j - M_i^* + \sum_{j=1}^{|H|} (N_H)_j \alpha_j - N_i \right\|_g \\
	= & ~  \left\| \sum_{j=1}^{|H|} (N_H)_j \alpha_j - N_i \right\|_g \\
	\leq & ~ \ati_{g,|H| + 1} \cdot \left( \| - N_i \|_g  + \sum_{j=1}^{|H|} \|  (N_H)_j \alpha_j \|_g \right) \\
	\leq & ~ \ati_{g,|H| + 1} \cdot \left( \| - N_i \|_g  + \mon_g\cdot\sum_{j=1}^{|H|} \| (N_H)_j \|_g \right) \\
	\leq & ~ \ati_{g,|H| + 1} \cdot \left( \sym_g \cdot \| N_i \|_g  + \mon_g\cdot\sum_{j=1}^{|H|} \| (N_H)_j \|_g \right) \\
	\leq & ~ \ati_{g,|H| + 1} \cdot \sym_g \cdot \left( \| N_i \|_g  + \sum_{j=1}^{|H|} \| (N_H)_j \|_g \right),
	\end{align*}
	where the second step follows from $M = M^* + N$, the third step follows from $ \sum_{j=1}^{|H|} (M_H^*)_j \alpha_j - M_i^* = 0$, the fourth step follows from the approximate triangle inequality, the fifth step follows from the fact that $|\alpha_j| \leq 1$ and $g$ is $\mon_g$-monotone, and the sixth step follows from that $g$ is $\mon_g$-monotone. 
	
\end{proof}

\subsection{Proof of Lemma~\ref{lem:pro_i_not_in_R_S_i}}

\begin{proof}
	Using Part (\RN{1}) of Lemma~\ref{lem:pro_random_S_i_two_parts}, we have
	\begin{align*}
	\Pr_{H\sim {{[m]}\choose{2k}},i\sim [m]\setminus H} [ i\not\in R_{M^*}(H\cup \{i\}) ] \geq 1/2. 
	\end{align*}
	For each set $H$, we define $P_H=\Pr_{i\sim [m]\setminus H}[ i\not\in R_{M^*}(H\cup \{i\}) ].$
	We have
	\begin{align}\label{eq:average_sum_n_choose_2k_is_at_least_1_over_2}
	1 \geq \frac{1}{{{m}\choose{2k}}}\sum_{H\in{{[m]}\choose{2k}}}P_H\geq 1/2.
	\end{align}
	We can show
	\begin{align*}
	& ~  \frac{1}{ {m \choose {2k}} } \left| \left\{ H ~\bigg|~ H\in{{[m]} \choose{2k}},P_H\geq 1/4 \right \} \right| \\
	= & ~ \frac{1}{ {m \choose {2k}} } \sum_{ H\in{{[m]}\choose{2k}},P_H\geq 1/4 } 1 \\
	\geq & ~  \frac{1}{ {m \choose {2k}} } \sum_{ H\in{{[m]}\choose{2k}},P_H\geq 1/4 } P_H \\
	\geq & ~ \frac{1}{2} - \frac{1}{ {m \choose {2k}} } \sum_{ H\in{{[m]}\choose{2k}},P_H < 1/4 } P_H \\
	\geq & ~ \frac{1}{2} - \frac{1}{ {m \choose {2k}} } \sum_{ H\in{{[m]}\choose{2k}},P_H < 1/4 } \frac{1}{4} \\
	\geq & ~ \frac{1}{2} - \frac{1}{ {m \choose {2k}} }  {m \choose {2k}} \frac{1}{4} \\
	= & ~ \frac{1}{4},
	\end{align*}
	where the second step follows since $1 \geq P_H$, the third step follows since Eq.~\eqref{eq:average_sum_n_choose_2k_is_at_least_1_over_2}, the fourth step follows since $P_H < 1/4$.
	
	Thus, we have
	\begin{align*}
	\left| \left \{H ~ \bigg| ~ H\in{{[m]}\choose{2k}},P_H\geq 1/4 \right\} \right|\geq {{m}\choose{2k}}/ 4 .
	\end{align*}
	Recall the definition of $P_H$, we have
	\begin{align*}
	\left| \left\{ H ~\bigg|~ H \in {[m] \choose 2k} , \Pr_{i\sim [m] \backslash H} [i\notin R_{M^*}(H \cup \{i\}) ] \geq 1/4 \right\} \right| \geq {m \choose 2k}/4,
	\end{align*}
	which implies
	\begin{align*}
	\Pr_{H \sim {{[m]}\choose{2k}}} \biggl[ \biggl| \biggl\{ i ~\bigg|~ i\in [m]\setminus H,i\not\in R_{M^*}(H\cup \{i\}) \biggr\} \biggr|\geq (m-2k)/4\biggr] \geq 1/4.
	\end{align*}
	
\end{proof}

\subsection{Proof of Claim~\ref{cla:t_i_greater_than_19_over_20}}

\begin{proof}
	For simplicity, we omit $i$ in all the subscripts in this proof.
	
	\begin{align*}
	& ~\Pr_{H\sim {T \choose 2k}} \left[ \sum_{j\in H} \| \Delta_j \|_g \leq 400 \frac{k}{m} \sum_{j= \frac{m}{100k} }^m \| \Delta_{t_j }\|_g \right] \\
	= & ~ \Pr_{H\sim {T \choose 2k}} \left[ \sum_{j\in H} \| \Delta_j \|_g \leq 400 \frac{k}{m} \sum_{j= \frac{m}{100k} }^m \| \Delta_{t_j }\|_g ~\bigg|~ \exists j \leq \frac{m}{100k}, t_j \in H \right] \cdot \Pr_{H\sim {T \choose 2k}} \left[ \exists j \leq \frac{m}{100k}, t_j \in H \right] \\
	+ & ~ \Pr_{H\sim {T \choose 2k}} \left[ \sum_{j\in H} \| \Delta_j \|_g \leq 400 \frac{k}{m} \sum_{j= \frac{m}{100k} }^m \| \Delta_{t_j }\|_g ~\bigg|~ \forall j \leq \frac{m}{100k}, t_j \notin H \right] \cdot \Pr_{H\sim {T \choose 2k}} \left[ \forall j \leq \frac{m}{100k}, t_j \notin H \right] \\
	\leq & ~ \underbrace{ \Pr_{H\sim {T \choose 2k}} \left[ \exists j \leq \frac{m}{100k}, t_j \in H \right] }_{C_1} + \underbrace{ \Pr_{H\sim {T \choose 2k}} \left[ \sum_{j\in H} \| \Delta_j \|_g \leq 400 \frac{k}{m} \sum_{j= \frac{m}{100k} }^m \| \Delta_{t_j }\|_g ~\bigg|~ \forall j \leq \frac{m}{100k}, t_j \notin H \right] }_{C_2}
	\end{align*}
	It remains to upper bound the terms $C_1$ and $C_2$. We can upper bound $C_1$:
	\begin{align*}
	C_1 = & ~ 1 - (1 - \frac{m/100k}{m} ) \cdot (1- \frac{m/100k}{m-1}) \cdot \cdots \cdot (1- \frac{m/100k}{m - 2k + 1}) \\
	\leq & ~ 1 - (1 - \frac{m/100k}{m/2} )^{2k} \\
	\leq & ~ 1 - (1 - \frac{1}{25} )\\
	= & \frac{1}{25},
	\end{align*}
	where the second step follows since $m \geq 1000 k$.
	
	Using Markov's inequality,
	\begin{align*}
	C_2 \leq \frac{ \underset{H\sim {T \choose 2k}}{\E} \left[ \sum_{j\in H} \| \Delta_j \|_g \leq 400 \frac{k}{m} \overset{m}{\underset{j= \frac{m}{100k} }{\sum}} \| \Delta_{t_j }\|_g ~\bigg|~ \forall j \leq \frac{m}{100k}, t_j \notin H \right] }{400\frac{k}{m} \overset{m}{\underset{j= \frac{m}{100k} }{\sum}} \| \Delta_{t_j} \|_g } \leq  1/100,
	\end{align*}
	where the second step follows since
	\begin{align*}
	& ~\E_{H\sim {T \choose 2k}} \left[ \sum_{j\in H} \| \Delta_j \|_g \leq 400 \frac{k}{m} \sum_{j= \frac{m}{100k} }^m \| \Delta_{t_j }\|_g ~\bigg|~ \forall j \leq \frac{m}{100k}, t_j \notin H \right] \\
	\leq & ~ \frac{2k}{m - m/100k} \sum_{j=\frac{m}{100k}}^m \| \Delta_{t_j} \|_g \\
	\leq & ~ 4 \frac{k}{m} \sum_{j=\frac{m}{100k} }^m \| \Delta_{t_j} \|_g
	\end{align*}
\end{proof}

\subsection{Proof of Claim~\ref{cla:t_i_less_than_m_over_5}}
\begin{proof} 
	For simplicity, we omit $i$ in all the subscripts in this proof.
	\begin{align*}
	& ~\left| \left\{ t_{j} ~\bigg|~ t_{j} \in T, \| \Delta_{t_{j}} \|_g \geq \frac{20}{m} \sum_{j'=\frac{m}{100k} }^{m} \| \Delta_{t_{j'}} \|_g \right\} \right| \\
	\leq & ~ \left| \left\{ t_{j} ~\bigg|~ t_{j} \in T, j \leq \frac{m}{100k} \right\} \right| + \left| \left\{ t_{j} ~\bigg|~ t_{j} \in T, j > \frac{m}{100k}, \| \Delta_{t_{j}} \|_g \geq \frac{20}{m} \sum_{j'=\frac{m}{100k}}^{m} \| \Delta_{t_{j'}} \|_g \right\} \right| \\
	\leq & ~ \left| \left\{ t_{j} ~\bigg|~ t_{j} \in T, j \leq \frac{m}{100k} \right\} \right| + \left| \left\{ t_{j} ~\bigg|~ t_{j} \in T, j > \frac{m}{100k}, \| \Delta_{t_{j}} \|_g \geq \frac{10}{m-\frac{m}{100k}} \sum_{j'=\frac{m}{100k}}^{m} \| \Delta_{t_{j'}} \|_g \right\} \right| \\
	\leq & ~ \frac{m}{100k} + \frac{1}{10} (m - \frac{m}{100k}) \\
	\leq & ~ \frac{1}{5} m,
	\end{align*}
	where the second step follows since $\frac{20}{m} \geq \frac{10}{m-m/100k}$
\end{proof}

\subsection{Proof of Lemma~\ref{lem:low_fitting_cost}}
\begin{proof}
	For simplicity, we omit $i$ in all the subscirptis in this proof.
	
	Let $M=A_{T}$, $M^*=A^*_{T}$ and $N=\Delta_T$.
	Then we can apply Lemma~\ref{lem:pro_i_not_in_R_S_i} and part (\RN{2}) of Lemma~\ref{lem:pro_random_S_i_two_parts}:
	\begin{align}\label{eq:random_S_i_cost_ati_sym_i_S_j}
	\Pr_{H \sim  {T \choose 2k} } \left[  \left| \left\{ j \in T ~ \bigg| ~ \min_{x\in \R^{|H|}} \| A_H x - A_j \|_g \leq \ati_{g,|H|+1} \cdot \sym_g \cdot \left( \| \Delta_j \|_g + \sum_{j'=1}^{|H|} \| (\Delta_H)_{j'} \|_g \right) \right\}  \right| \geq \frac{m}{4}  \right] \geq \frac{1}{4}
	\end{align}
	By Claim~\ref{cla:t_i_greater_than_19_over_20}, we have
	\begin{align}\label{eq:bound_sum_Delta_S_j}
	\Pr_{H \sim  {T \choose 2k} } \left[ \sum_{j=1}^{|H|} \| (\Delta_H)_j \|_g \leq 400 \frac{k}{m} \sum_{j=\frac{m}{100k}}^m \| \Delta_{t_j} \|_g \right] \geq \frac{19}{20}
	\end{align}
	Due to Claim~\ref{cla:t_i_less_than_m_over_5},
	\begin{align*}
	\left| \left\{ t_j ~\bigg|~ t_j \in T, \| \Delta_{t_j} \|_g \geq \frac{20}{m} \sum_{j'=\frac{m}{100k}}^m \| \Delta_{t_{j'}} \|_g \right\} \right| \leq \frac{1}{5} m
	\end{align*}
	Combining the above equation with the pigeonhole principle, for any $I \subseteq T$ with $|I| \geq m/4$, we have
	\begin{align}\label{eq:bound_Delta_i}
	\left| \left\{ t_j ~\bigg|~ t_j \in I, \| \Delta_{t_j} \|_g < \frac{20}{m} \sum_{j'=\frac{m}{100k}}^m \| \Delta_{t_{j'}} \|_g \right\} \right| \geq \frac{1}{4} m - \frac{1}{5} m = \frac{1}{20} m
	\end{align}
	Consider the quantity $\| \Delta_j \|_g + \sum_{j=1}^{|H|} \| (\Delta_H)_{j'} \|_g $ in Eq.~\eqref{eq:random_S_i_cost_ati_sym_i_S_j}. We use Eq.~\eqref{eq:bound_sum_Delta_S_j} and Eq.~\eqref{eq:bound_Delta_i} to provide an upper bound,
	\begin{align*}
	\| \Delta_j \|_g + \sum_{j'=1}^{|H|} \| (\Delta_H)_{j'} \|_g \leq \left( \frac{20}{m} + \frac{400k}{m} \right) \sum_{j' = \frac{m}{100k}}^m \|\Delta_{t_{j'}} \|_g.
	\end{align*}
	Eq.~\eqref{eq:bound_sum_Delta_S_j} will decrease the final probability by $(1-19/20)$ (from $1/4$ to $1/4 - 1/20$). Eq.~\eqref{eq:bound_Delta_i} will decrease the size of this set of $j$ by $\frac{1}{5} m$ (from $\frac{1}{4} m$ to $\frac{1}{4} m - \frac{1}{5} m$).
	
	Putting it all together, we can update Eq.~\eqref{eq:random_S_i_cost_ati_sym_i_S_j} in the following sense,
	\begin{align*}
	\Pr_{H \sim {T \choose 2k}} \left[ \left| \left\{ j ~\bigg|~ j\in T, \min_{x\in \R^{|H|}}\| A_H x - A_j \|_g \leq C\cdot\frac{1}{m} \cdot \sum_{j'=\frac{m}{100k} }^m \| \Delta_{t_{j'}} \|_g \right\} \right|  \geq (\frac{1}{4} -\frac{1}{5}) m \right]
	\geq \frac{1}{4}-\frac{1}{20}.
	\end{align*}
	where
	\begin{align*}
	C = (400+20) \cdot k \cdot \ati_{g,|H|+1} \cdot \sym_g.
	\end{align*}
\end{proof}

\subsection{Proof of Theorem~\ref{thm:intro_general_function_low_rank}}

\begin{proof}
	The running time is discussed at the beginning of Section~\ref{sec:alg}.
	In the remaining of the proof, we will focus on the correctness of Algorithm~\ref{alg:general_function_low_rank}.
	
	Firstly, let us consider the size of the output $S$. 
	For $i\in\{0\}\cup[r]$, let $m_i=|T_i|$.
	We set number of rounds $r$ to be the smallest value such that $m_r<1000 k$. 
	By the algorithm, we have $m_i=m_{i-1}-2k-(m_{i-1}-2k)/20\leq 19/20\cdot m_{i-1}$.
	Thus, $r=O(\log n)$.
	In each round $i$, the size of $S_i$ is $2k$.
	Then $|S|=|T_r|+\sum_{i=1}^r |S_i|\leq 1000k+r\cdot 2k\leq O(k\log n)$.
	
	Next, let us consider the quality of $S$.
	Since each regression call has $1-1/\poly(n)$ success probability, all the regression calls succeed with probability at least $1-1/\poly(n)$.
	In the remaining of the proof, we condition on that all the regression calls succeed.
	
	Let us fix $i\in[r],j\in[\log n]$.
	Recall that $T_i=\{t_{i,1},t_{i,2},\cdots,t_{i,m_i}\}$ and $\|\Delta_{t_{i,1}}\|_g\geq \|\Delta_{t_{i,2}}\|_g\geq \cdots\geq \|\Delta_{t_{i,m_i}}\|_g$.
	By regression property and Lemma~\ref{lem:low_fitting_cost}, with probability at least $1/4$,
	\begin{align*}
	&\sum_{q\in R_i^{(j)}}\min_{x\in\mathbb{R}^{2k}} \|A_{S^{(j)}_i}x-A_q\|_g\\
	\leq~&
	\sum_{q \in R_{i}^{(j)}} v_{i,q}^{(j)}\\
	 \leq~& \reg_{g,2k}\cdot (m_i-2k)\cdot 500\cdot k\cdot \ati_{g,2k+1}\cdot \mon_g/m_i\cdot \sum_{j'=\frac{m_i}{100k}}^{m_i} \|\Delta_{t_{i,j'}}\|_g\\
	 \leq~& \reg_{g,2k+1}\cdot \ati_{g,2k+1}\cdot \mon_g\cdot O(k)\cdot \sum_{j'=\frac{m_i}{100k}}^{m_i} \|\Delta_{t_{i,j'}}\|_g.
	\end{align*}
	For each $i\in[r]$, since we repeat $\log (n)$ times, the success probability can be boosted to at least $1-1/\poly(r)$, i.e., with probability at least $1-1/\poly(r)$, we have
	\begin{align}\label{eq:small_fitting_cost}
	&\sum_{q\in R_i}\min_{x\in\mathbb{R}^{2k}} \|A_{S_i}x-A_q\|_g
	\leq~& \reg_{g,2k+1}\cdot \ati_{g,2k+1}\cdot \mon_g\cdot O(k)\cdot \sum_{j'=\frac{m_i}{100k}}^{m_i} \|\Delta_{t_{i,j'}}\|_g.
	\end{align}
	In the remaining of the proof, we condition on above inequality for every $i\in[r]$.
	Without loss of generality, we suppose $\|\Delta_1\|_g\geq \|\Delta_2\|_g\geq \cdots \geq \|\Delta_n\|_g$.
	We have 
	\begin{align*}
	&\sum_{q=1}^n \min_{x\in\mathbb{R}^{|S|}}\|A_Sx-A_q\|_g\\
	\leq~& \left(\sum_{q\in T_r}\min_{x\in\mathbb{R}^{|T_r|}} \|A_{T_r}x-A_q\|_g\right)+\left(\sum_{i=1}^r \sum_{q\in T_{i-1}\setminus T_i} \min_{x\in\mathbb{R}^{2k}}\|A_{S_i}x-A_q\|_g\right) \\
	=~ & \sum_{i=1}^r \sum_{q\in T_{i-1}\setminus T_i} \min_{x\in\mathbb{R}^{2k}}\|A_{S_i}x-A_q\|_g\\
	\leq~& \sum_{i=1}^r  \reg_{g,2k+1}\cdot \ati_{g,2k+1}\cdot \mon_g\cdot O(k)\cdot \sum_{j'=\frac{m_i}{100k}}^{m_i} \|\Delta_{t_{i,j'}}\|_g\\
	\leq~& \sum_{i=1}^r \reg_{g,2k+1}\cdot \ati_{g,2k+1}\cdot \mon_g\cdot O(k)\cdot \sum_{j'=\frac{m_i}{100k}}^{m_i} \|\Delta_{j'}\|_g\\
	= ~&  \reg_{g,2k+1}\cdot \ati_{g,2k+1}\cdot \mon_g\cdot O(k)\cdot  \sum_{j'= 1 }^{n} \| \Delta_{j'} \|_g   \left( \argmin_{i \in [r] } \left\{ m_i <j' \right\}  - \argmin_{i\in[r]} \left\{ \frac{ m_i }{ 100k } < j' \right\} + O(1) \right) \\
	\leq~ & \reg_{g,2k+1}\cdot \ati_{g,2k+1}\cdot \mon_g\cdot O(k\log k) \cdot \|\Delta\|_g,
	\end{align*}	
	where the third step follows from $T_{i-1}\setminus T_i = S_i\cup R_i$ and Equation~\eqref{eq:small_fitting_cost}, the forth step follows from $\|\Delta_{j'}\|_g\geq \|\Delta_{t_{i,j'}}\|_g$, and the last step follows from $\left( \argmin_{i \in [r] } \left\{ m_i <j' \right\}  - \argmin_{i\in[r]} \left\{ \frac{ m_i }{ 100k } < j' \right\} + O(1) \right)\leq O(\log k)$.

\end{proof}

\section{Necessity of the Properties of $g$}\label{sec:necessity}

We note that an approximate triangle inequality is necessary to obtain
a column subset selection algorithm. An example function not satisfying
this is the ``jumping function'': $g_{\tau}(x) = |x|$ if $|x| \geq \tau$,
and $g_{\tau}(x) = 0$ otherwise. For the identity matrix $I$
and any $k = \Omega(\log n)$, the Johnson-Lindenstrauss lemma implies
one can find a rank-$k$
matrix $B$ for which $\|I-B\|_{\infty} < 1/2$, that is, all entries
of $I-B$ are at most $1/2$. If we set $\tau = 1/2$, then
$\|I-B\|_{g_{\tau}} = 0$, but for any subset $I_S$ of columns of the identity
matrix we choose, necessarily $\|I-I_SX\|_{\infty} \geq 1$,
so $\|I-B\|_{g_{\tau}} > 0$. Consequently, there is no subset of
a small number of columns which obtains a $\poly(k \log n)$-approximation
with the jumping function loss measure.

While the jumping function does not satisfy the Approximate triangle
inequality, it does satisfy our only other required
structural property, the Monotone
property.

There are interesting examples of functions $g$ which are only approximately
monotone in the above sense, such as the quantile function $\rho_{\tau}(x)$,
studied in \cite{ymm14} in the context of regression, where for a given
parameter
$\tau$, $\rho_{\tau}(x) = \tau x$ if $x \geq 0$, and
$\rho_{\tau}(x) = (\tau -1)x$ if $x < 0$. Only when $\tau = 1/2$
is this a monotone function with
$\sym_g = 1$ in the above definition, in which case
it coincides with the absolute value function up to a factor of $1/2$. For
other constant $\tau \in (0,1)$, $\sym_g$ is a constant. The loss
function $\rho_{\tau}(x)$ is also sometimes
called the scalene loss, and studied in the context of low rank approximation
in \cite{uhzb16}. 

When $\tau = 1$ this is the so-called Rectified Linear Unit (ReLU)
function in machine learning, i.e., $\rho_1(x) = x$ if $x \geq 0$ and
$\rho_1(x) = 0$ if $x < 0$. In this case $\sym_g = \infty$.
and the optimal rank-$k$ approximation for
any matrix $A$ is $0$, since $\|A-\lambda {\bf 1} {\bf 1}^\top \|_{\rho_1} = 0$ if one
sets $\lambda$ to be a large enough positive number, thereby making
all entries of $A-\lambda {\bf 1} {\bf 1}^\top$ negative and their corresponding cost equal
to $0$. Notice though, that there are no good column subset selection
algorithms for some matrices $A$, such as the $n \times n$ identity matrix.
Indeed, for the identity, if we choose any subset $A_S$
of at most $n-1$ columns of $A$,
then for any matrix $X$ there will be an entry of $A-A_SX$ which is positive,
causing the cost to be positive. Since we will restrict ourselves to
column subset selection, being approximately monotone with
a small value of $\sym_g$ in the above definition is in fact necessary to
obtain a good approximation with a small number of columns,
as the ReLU function illustrates (see also related functions such as the
leaky ReLU and squared ReLU \cite{zsjbd17,bhl18,gkm18}).

Note that the ReLU function is an example which satisfies the triangle
inequality, showing that our additional assumption of approximate monotonicity
is required.

Thus, if either property fails to hold, there need not be a small subset
of columns spanning a relative error approximation.
These examples are stated in more detail below.

\subsection{Functions without Approximate Triangle Inequality}\label{sec:no_ati} 

In this section, we show how to construct a function $f$ such that it is not possible to obtain a good entrywise-$f$ low rank approximation by selecting a small subset of columns. Furthermore, $f$ is monotone but does not have the approximate triangle inequality. Theorem~\ref{thm:f_without_triangle_inequality_is_no_good_subset} shows this result.

First, we show that a small subset of columns cannot give a good low rank approximation in $\ell_{\infty}$ norm.
Then we reduce the $\ell_{\infty}$ column subset selection problem to the entrywise-$f$ column subset selection problem.

The following is the Johnson-Lindenstrauss lemma.
\begin{lemma}[JL Lemma]\label{lem:easy_to_approximate_identity}
For any $n\geq 1,\varepsilon\in(1/\sqrt{n},1/2),$ there exists $U\in\mathbb{R}^{n\times k}$ with $k=O(\varepsilon^{-2}\log(n))$ such that $\|UU^{\top}-I_n\|_{\infty}\leq O(\varepsilon),$ where $I_n\in\mathbb{R}^{n\times n}$ is an identity matrix.
\end{lemma}

\begin{theorem}\label{thm:ell_infty_hard}
For $n\geq 1,$ there is a matrix $A\in \R^{n \times n}$ with the following properties. Let $k = \Theta(\varepsilon^{-2}\log(n))$ for an arbitrary $\varepsilon\in(1/\sqrt{n},1/2)$. Let $D\in \R^{n \times n}$ denote a diagonal matrix with $n-1$ nonzeros on the diagonal. We have
\begin{align*}
\min_{X \in \R^{n\times n}} \| X D A - A \|_{\infty} \geq 1
\end{align*}
and
\begin{align*}
\min_{\rank-k~A'} \| A' - A \|_{\infty}<O(\varepsilon).
\end{align*}
\end{theorem}
\begin{proof}
We choose $A$ to be the identity matrix. 
By Lemma~\ref{lem:easy_to_approximate_identity}, we can find a rank-$k$ matrix $B$ for which 
\begin{align*}
 \|A-B\|_{\infty}\leq O(\varepsilon).
\end{align*}

Since $A$ is an $n\times n$ identity matrix, even if we can use $n-1$ columns to fit the other columns, the cost is still at least $1$. 

\end{proof}

In the following, we state the construction of our function $f$.
\begin{definition}\label{def:hardness_jumping_function}
We define function $f(x)$ to be $f(x) = c$ if $|x| > \tau$ and $f(x) = 0$ if $|x|\leq \tau$. Given matrix $A$, we define $\| A\|_f = \sum_{i=1}^{n} \sum_{j=1}^{n} f(A_{i,j})$.
\end{definition}

\begin{theorem}[No good subset of columns]\label{thm:f_without_triangle_inequality_is_no_good_subset}
For any $n\geq 1$, there is a matrix $A\in \R^{n \times n}$ with the following property. Let $k\geq c\log n $ for a sufficiently large constant $c>0$. Let $D\in \R^{n \times n}$ denote an arbitrary diagonal matrix with $n-1$ nonzeros on the diagonal. For $f$ with 
parameter $\tau=1/4,$  we have
\begin{align*}
\min_{X \in \R^{n\times n}} \| X D A - A \|_f >0
\end{align*}
and
\begin{align*}
\min_{\rank-k~A'} \| A' - A \|_{f}=0.
\end{align*}
\end{theorem}
\begin{proof}
We can set $A$ to be the identity matrix.
Due to Theorem~\ref{thm:ell_infty_hard}, there exists $A'$ for which $\min_{\rank-k~A'} \| A' - A \|_{\infty}<1/4$, which implies that $\min_{\rank-k~A'} \| A' - A \|_{f}=0.$
Also due to Theorem~\ref{thm:ell_infty_hard}, we have $\min_{X \in \R^{n\times n}} \| X D A - A \|_\infty=1,$ and thus, 
$
\min_{X \in \R^{n\times n}} \| X D A - A \|_f >0.
$

\end{proof}

\subsection{$\relu$ Function Low Rank Approximation}\label{sec:no_mon}
In this section, we discuss a function which has the approximate triangle inequality but is not monotone.
The specific function we discuss in this section is $\relu.$
The definition of $\relu$ is defined in Definition~\ref{def:hardness_relu_function}.
First, we show that $\relu$ low rank approximation has a trivial best rank-$k$ approximation.
Second, we show that for some matrices, there is no small subset of columns which can give a good low rank approximation.

\begin{definition}\label{def:hardness_relu_function}
We define function $\relu(x)$ to be $\relu(x) = \max(0,x)$. Given matrix $A$, we define $\| A\|_{\relu} = \sum_{i=1}^{n} \sum_{j=1}^{n} \relu(A_{i,j})$.
\end{definition}

In the rank-$k$ approximation problem, given an input matrix $A$, the goal is to find a rank-$k$ matrix $B$ for which $\|A-B\|_{\relu}$ is minimized. A simple observation is that if we set $B$ to be a matrix with each entry of value $\|A\|_{\infty},$ then the value of each entry of $A-B$ is at most $0$. Thus, $\|A-B\|_{\relu}=0.$ Furthermore, the rank of $B$ is $1$.

Now, consider the column subset selection problem, let input matrix $A\in\mathbb{R}^{n\times n}$ be an identity matrix.
Then even if we can choose $n-1$ columns, they can never fit the remaining column.
Thus, the cost is at least $1$.
But as discussed, the best rank-$k$ cost is always $0$.
This implies that any subset of columns cannot give a good rank-$k$ approximation.

\section{Regression Solvers}\label{sec:regression}

In this section, we discuss several regression solvers.

\subsection{Regression for Convex $g$}
Notice that when the function $g$ is convex, the regression problem $\min_{X\in \mathbb{R}^{d\times m}}\|AX-B\|_g$ for any given matrices $A\in\mathbb{R}^{n\times d},B\in\mathbb{R}^{n\times m}$ is a convex optimization problem. Thus, it can be solved exactly by convex optimization algorithms.
\begin{fact}
Let $g$ be a convex function. Given $A\in\mathbb{R}^{n\times d},B\in\mathbb{R}^{n\times m},$ the regression problem $\min_{X\in\mathbb{R}^{d\times m}} \|AX-B\|_g$ can be solved exactly by convex optimization in $\poly(n,d,m)$ time.
\end{fact}
If a function $g$ has additional properties, i.e. 
$g$ is symmetric, monotone and grows subquadratically,
then there is a better running time constant approximation algorithm shown in~\cite{cw15soda}. 
Here ``grows quadratically'' means that there is an $\alpha\in [1,2]$ and $c_g>0$ so that for $a,a'$ with $|a|>|a'|>0,$
\begin{align*}
\left|\frac{a}{a'}\right|^{\alpha}\geq \frac{g(a)}{g(a')}\geq c_g \left|\frac{a}{a'}\right|.
\end{align*}
This kind of function $g$ is also called a ``sketchable'' function.
Notice that the Huber function satisfies the above properties.

\begin{theorem}[Modified version of Theorem 3.1 of \cite{cw15soda}]\label{thm:general_regression}
 Function $g$ is symmetric, monotone and grows subquadratically ($g$ is a $G$-function defined by \cite{cw15soda}). Given a matrix $A\in \mathbb{R}^{n\times d}$ and a matrix $B\in\mathbb{R}^{n\times m},$
 there is an algorithm which can output a matrix $\wh{X}\in\mathbb{R}^{d\times m}$ and a fitting cost vector $y\in\mathbb{R}^m$ such that with probability at least $1-1/\poly(nm),$
 $\forall i\in[m],\|A\wh{X}_i-B_i\|_g\leq O(1)\cdot \min_{x\in\mathbb{R}^d}\|Ax-B_i\|_g,$ and $y_i=\Theta(\|A\wh{X}_i-B_i\|_g).$ Furthermore, the running time is at most $\wt{O}(\nnz(A)+\nnz(B)+m\cdot\poly(d\log n))$.
\end{theorem}
\begin{proof}
We run $O(\log (nm))$ repetitions of the single column regression algorithm shown in Theorem 3.1 of \cite{cw15soda} for all columns $B_i$ for $i\in[m]$.
For each regression problem $\|Ax-B_i\|_g,$ we take the solution whose estimated cost is the median among these $O(\log (nm)$ repetitions as $\wh{X}_i$. 
Then by the Chernoff bound, we can boost the success probability of each column to $1-1/\poly(nm).$
By taking a union bound over all columns, we complete the proof.
\end{proof}

\subsection{$\ell_p$ Regression}
One of the most important cases in regression and low rank approximation problems is when the error measure is $\ell_p$.
For $\ell_p$ regression, though it can be solved by convex optimization/linear programming exactly, we can get a much faster running time if we allow some approximation ratios. In the following theorem, we show that there is an algorithm which can be used to solve $\ell_p$ regression for any $p\geq 1.$

\begin{theorem}[Modified version of~\cite{wz13}]\label{thm:l1_regression}
Let $p\geq 1,\varepsilon\in(0,1).$ Given a matrix $A\in \mathbb{R}^{n\times d}$ and a matrix $B\in\mathbb{R}^{n\times m},$ there is an algorithm which can output a matrix $\wh{X}\in\mathbb{R}^{d\times m}$ and a fitting cost vector $y\in\mathbb{R}^m$ such that with probability at least $1-1/\poly(nm),$
 $\forall i\in[m],\|A\wh{X}_i-B_i\|_p^p\leq (1+\varepsilon)\cdot \min_{x\in\mathbb{R}^d}\|Ax-B_i\|_p^p,$ and $y_i=\Theta(\|A\wh{X}_i-B_i\|_p^p).$ Furthermore, the running time is at most $\wt{O}(\nnz(A)+\nnz(B)+mn^{\max(1-2/p,0)}\cdot\poly(d))$.
\end{theorem}
\begin{proof}
As in the proof of Theorem~\ref{thm:general_regression}, we only need to run $O(\log (nm))$ repetitions of the single column regression algorithm shown in~\cite{wz13}.
\end{proof}

\subsection{$\ell_0$ Regression}

\begin{definition}[Regular partition]\label{def:l0_regular_partition}
Given a matrix $A\in \R^{n \times k}$, we say $\{ S_1, S_2, \cdots, S_h\}$ is a regular partition for $[n]$ with respect to the matrix $A$ if, for each $i \in [h]$,
\begin{align*}
\rank( A^{S_i} ) = |S_i|, \text{~and~} \mathrm{rowspan}(A^{S_i}) = \mathrm{rowspan} \left( A^{\cup_{j=i}^h S_j }  \right),
\end{align*}
where $A^{S_i} \in \R^{|S_i| \times k}$ denotes the matrix that selects a subset $S_i$ of rows of the matrix $A$.
\end{definition}

\begin{algorithm}
\begin{algorithmic}\caption{$\ell_0$ regression \cite{apy09}}\label{alg:l0_regression}
\Procedure{\textsc{L0Regression}}{$A,b,n,k,c$} \Comment{Theorem~\ref{thm:l0_regression}}
	\State $x' \leftarrow 0^k$
	\State $\{S_1,S_2, \cdots, S_h\} \leftarrow \textsc{GenerateRegularPartition}(A,n,k)$
	\State $x' \leftarrow 0^k$
	\For{$i=1 \to h$}
		\State Find a $\wt{x} $ such that $A^{S_i} \wt{x} = b_{S_i}$
 		\If{$\| A\wt{x} - b \|_0 < \| A x' - b\|_0 $}
			\State $x' \leftarrow \wt{x}$
		\EndIf
	\EndFor
	\State \Return $x'$
\EndProcedure
\end{algorithmic}
\end{algorithm}

\cite{apy09} studied the Nearest Codeword problem over finite fields $\mathbb{F}_2$. Their proof can be extended to the real field and generalized to Theorem~\ref{thm:l0_regression}. For completeness, we still provide the proof of the following result.
\begin{theorem}[Generalization of \cite{apy09}]\label{thm:l0_regression}
Given matrix $A\in \R^{n \times k}$ and vector $\R^n$, for any $c\in [1,k]$, there is an algorithm (Algorithm~\ref{alg:l0_regression}) that runs in $n^{O(1)}$ time and outputs a vector $x'\in \R^k$ such that
\begin{align*}
\| A x' - b \|_0 \leq k  \min_{x \in \R^k } \| A x - b \|_0.
\end{align*}
\end{theorem}

\begin{proof}
Let $x^* \in \R^k$ denote the optimal solution to $\min_{x\in \R^k} \| A x - b \|_0$. We define set $E$ as follows
\begin{align*}
E = \{ i \in [n] ~ | ~ (Ax^*)_i \neq b_i \}.
\end{align*}
We create a regular partition $\{S_1,S_2,\cdots, S_h\}$ for $[n]$ with respect to $A$.

Let $i$ denote the smallest index such that $|S_i \cap E| = 0$, i.e.,
\begin{align*}
i = \min \{ j ~|~ |S_j \cap E| = 0 \}.
\end{align*}
The linear equation we want to solve is $A^{S_i } x = b_{S_i}$. Let $\wt{x} \in \R^k$ denote a solution to $A^{S_i} \wt{x} = A^{S_i} x^*$ (Note that, by our choice of $i$, $b_{S_i} = A^{S_i}x^*$). Then we can rewrite $\| A \wt{x} -b \|_0$ in the following sense,
\begin{align}\label{eq:l0_rewrite_Awtx_minus_b}
\| A \wt{x} - b \|_0 = \sum_{j=1}^{i-1} \left\| A^{S_j} \wt{x} - b_{S_j} \right\|_0 + \sum_{j=i}^h \left\| A^{S_j} \wt{x} - b_{S_j} \right\|_0.
\end{align}

For each $j \in \{1,2,\cdots, i-1\}$, we have
\begin{align}\label{eq:l0_rewrite_Awtx_minus_b_part1}
\| A^{S_j} \wt{x} - b_{S_j} \|_0 \leq & ~ k \notag \\
\leq & ~ \| A^{S_j} x^* - b_{S_j} \|_0 \cdot \lceil k  \rceil,
\end{align}
where the first step follows from $|S_0|\leq k$, and the last step follows from $\| A^{S_j} x^* - b_{S_j} \|_0 \geq 1$, $\forall j \in [i-1]$.

Note that, by our choice of $i$, we have $A^{S_i} \wt{x} = A^{S_i} x^*$. Then for each $j \in \{i,i+1, \cdots, n\}$, using the regular partition property, there always exists a matrix $P_{(j)}$ such that $A^{S_j} = P_{(j)} A^{S_i}$. Then we have
\begin{align}\label{eq:l0_rewrite_Awtx_minus_b_part2}
A^{S_j} \wt{x} = P_{(j)} A^{S_i} \wt{x} = P_{(j)} A^{S_i} x^* = A^{S_j} x^*.
\end{align}

Plugging Eq.~\eqref{eq:l0_rewrite_Awtx_minus_b_part1} and \eqref{eq:l0_rewrite_Awtx_minus_b_part2} into Eq.~\eqref{eq:l0_rewrite_Awtx_minus_b}, we have
\begin{align*}
\| A \wt{x} - b \|_0 = & ~ \sum_{j=1}^{i-1} \left\| A^{S_j} \wt{x} - b_{S_j} \right\|_0 + \sum_{j=i}^h \left\| A^{S_j} \wt{x} - b_{S_j} \right\|_0 \\
\leq & ~ k \sum_{j=1}^{i-1} \left\| A^{S_j} x^* - b_{S_j} \right\|_0 + \sum_{j=i}^h \left\| A^{S_j} x^* - b_{S_j} \right\|_0 \\
\leq & ~ k  \| A x^* - b \|_0.
\end{align*}
This completes the proof.
\end{proof}

\section{Hardness}\label{sec:hardinstance}

\subsection{Column Subset Selection for the Huber Function}

The rough idea here is to define $k = \Omega(\sqrt{\log n})$ groups of columns, where
we carefully choose the $i$-th group to have $n^{1-2i\epsilon}$ columns, $\eps = .2/(1.5k)$, and
in the $i$-th group each column has the form
$$n^{1.5i\epsilon} \cdot 1^n + [\pm n^{-.2+i\epsilon}, \ldots, \pm n^{-.2+i \epsilon}, \pm n^{.5+2i\epsilon}, \ldots, \pm n^{.5+2i\epsilon}],$$
where there are $n-n^{.1}$ coordinates where the perturbation is randomly either $+n^{-.2+i\epsilon}$ or $-n^{-.2+i\epsilon}$,
and the remaining $n^{.1}$ coordinates are randomly either $+n^{.5+2i\epsilon}$ or $-n^{.5+2i\epsilon}$. We call the former
type of coordinates ``small noise'', and the latter ``large noise''.
All remaining columns
in the matrix are set to $0$.
Because of the random signs,
it is very hard to fit the noise in one column to that of another column.
One can show, that to approximate a column in the $j$-th group by a column in the $i$-th group, $i < j$, one needs to scale
by roughly $n^{1.5(j-i)\epsilon}$, just to cancel out the ``mean'' $n^{1.5j\epsilon} \cdot 1^n$. But when doing so, since the Huber
function is quadratic for small values, the scaled small noise is now magnified more than linearly compared to what it was
before, and this causes a column in the $i$-th group not to be a good approximation of a column in the $j$-th group. On the
other hand, if you want to approximate a column in the $j$-th group by a column in the $i$-th group, $i > j$, one again
needs to scale by roughly $n^{1.5(j-i)\epsilon}$ just to cancel out the ``mean'', but now one can show
the large noise from the column in the $i$-th group is too large and remains in the linear regime,
causing a poor approximation. The details of this construction are given in the following theorem.

\begin{theorem}
Let $H(x)$ denote the modified Huber function with $\tau =1$, i.e.,
\begin{align*}
H (x) =
\begin{cases}
x^2 /\tau, & \text{~if~} |x| < \tau;\\
|x|, & \text{~if~} |x| \geq \tau.
\end{cases}
\end{align*}
For any $n\geq 1,$ 
there is a matrix $A\in \R^{n \times n}$ such that, if we select 
$o(\sqrt{\log n})$ columns to fit the entire matrix, there is 
no $O(1)$-approximation, i.e., for any subset $S\subseteq[n]$ with $|S|=o(\sqrt{\log n}),$
\begin{align*}
\min_{X \in \R^{|S| \times n} } \| A_S X -A \|_H \geq 
\omega(1)\cdot\min_{\rank-1~A'} \| A' - A \|_H.
\end{align*}
\end{theorem}
\begin{proof}
Suppose there is an algorithm which only finds a subset with size $k/2=o(\sqrt{\log n}).$ We want to prove a lower bound on its approximation ratio.

Let $\varepsilon=0.2/(1.5k).$
Let $A$ denote a matrix with $k+1$ groups of columns.

For each group $i\in [k]$, $I_i$ has $n^{1-2i\epsilon}$ columns which are
\begin{align*}
\begin{bmatrix}
n^{1.5i\epsilon} \\
n^{1.5i\epsilon} \\
n^{1.5i\epsilon}\\
\vdots \\
n^{1.5i\epsilon}\\
n^{1.5i\epsilon}\\
n^{1.5i\epsilon} \\
n^{1.5i\epsilon}
\end{bmatrix}
+
\begin{bmatrix}
\pm n^{-0.2+ i \epsilon} \\
\pm n^{-0.2+ i \epsilon} \\
\vdots \\
\pm n^{-0.2+ i \epsilon} \\
\pm n^{0.5 + 2 i \epsilon}\\
\pm n^{0.5 + 2 i \epsilon}\\
\vdots\\
\pm n^{0.5 + 2 i \epsilon}
\end{bmatrix}
\in \R^n,
\end{align*}
where $\pm$ indicates i.i.d. random signs. For the error column, the first $n-n^{0.1}$ rows are $n^{-0.2+ i \epsilon}$, and the last $n^{0.1}$ rows are $n^{0.5 + 2 i \epsilon}.$

The last group of $n-\sum_{i=1}^k n^{1-2i\epsilon}$ columns are 
\begin{align*}
\begin{bmatrix}
0 \\
0 \\
\vdots \\
0 \\
0
\end{bmatrix}
\in \R^n.
\end{align*}
The optimal cost is at most
\begin{align*}
\sum_{i=1}^k n^{1-2 i \epsilon} \cdot ( (n-n^{0.1}) H( n^{-0.2+i\epsilon} ) + n^{0.1}H(n^{0.5+2i\epsilon}) )
\leq \sum_{i=1}^k n^{1-2 i \epsilon} ( n \cdot n^{-0.4+2i\epsilon} + n^{0.6+2i\epsilon} )
\leq O( k n^{1.6}).
\end{align*}
where the second step follows since $n^{-0.2+i\epsilon} < 1$ and $n^{0.5+2i\epsilon} \geq 1$. Thus, it implies
\begin{align*}
\min_{\rank-1~A'} \|A' - A \|_H\leq O (k n^{1.6}).
\end{align*}

Now let us consider the lower bound for using a subset of columns to fit the matrix. 
First, we fix a set $S=\{j_1,j_2,\cdots,j_{k/2}\}$ of $k/2$ columns.
Since there are $k$ groups, and $|S|\leq k/2,$ the number of groups $I_i$ for $i\in[k]$ with $S\cap I_i=\emptyset$ is at least $k/2.$
It means that there are at least $k/2$ groups for which $S$ does not have any column from them.
Notice that the optimal cost is at most $O(kn^{1.6})$, so it suffices to prove that $\forall i\in[k]$ with $I_i\cap S=\emptyset,$ each column $j\in I_i$ will contribute a cost of $\omega(n^{0.6+2i\varepsilon}).$

For notation, we use $\group(j)$ to denote the index of the group which contains the column $j$.
For each column $j$, we use $\Delta_j$ to denote the noise part, and use $A^*_j$ to denote the rank-$1$ ``ground truth'' part.
Notice that $A_j=A^*_j+\Delta_j.$

\begin{claim}[Noise cannot be used to fit other vectors]\label{cla:noise_cannot_fit_noise}
Let $x_1,x_2,\cdots,x_s\in\mathbb{R}^{m}$ be $s$ random sign vectors.
Then with probability at least $1-2^s\cdot 2^{-\Theta(m/2^s)},$ $\forall \alpha_1,\alpha_2,\cdots,\alpha_s\in \mathbb{R},$ the size of $Z_{\alpha_1,\alpha_2,\cdots,\alpha_s}=\{i\in [m]\mid \sign((\alpha_1x_1)_i)=\sign((\alpha_2x_2)_i)=\cdots=\sign((\alpha_sx_s)_i)=\sign(+1)\}$ is at least $\Omega(m/2^s).$
\end{claim}
\begin{proof}
For a set of fixed $\alpha_1,\alpha_2,\cdots,\alpha_s,$ the claim follows from the Chernoff bound.
Since there are $2^s$ different possibilities of signs of $\alpha_1,\cdots,\alpha_s,$ taking 
a union bound over them completes the proof.
\end{proof}

Now we consider a specific column $j\in I_i$ for some $i\in[k],$ where $I_i\cap S=\emptyset.$
Suppose the fitting coefficients are $\alpha_1,\alpha_2,\cdots,\alpha_{k/2}.$
Consider the following term
\begin{align*}
&(\alpha_1A_{j_1}+\alpha_2A_{j_2}+\cdots+\alpha_{k/2}A_{k/2})-A_j\\
=~&(\alpha_1A^*_{j_1}+\alpha_2A^*_{j_2}+\cdots+\alpha_{k/2}A^*_{k/2}-A^*_j)+(\alpha_1\Delta_{j_1}+\alpha_2\Delta_{j_2}+\cdots+\alpha_{k/2}\Delta_{k/2}-\Delta_j)
\end{align*}
Let $u^*=(\alpha_1A^*_{j_1}+\alpha_2A^*_{j_2}+\cdots+\alpha_{k/2}A^*_{k/2}-A^*_j).$
By Claim~\ref{cla:noise_cannot_fit_noise}, with probability at least $1-(k/2)\cdot 2^{-\Theta(n/2^{k/2})},$
\begin{align*}
|\{t\in[n]\mid \sign(u^*_t)=\sign(\alpha_1\Delta_{j_1,t})=\cdots=\sign(\alpha_{k/2}\Delta_{j_{k/2},t})=\sign(-\Delta_{j,t})\}|=\Omega(n/2^{k/2}).
\end{align*}
Observe that all the coordinates of $u^*$ are the same, 
and the absolute value of each entry of $u^*$ should be at most $O(n^{1.5i\varepsilon}).$
Otherwise the column already has $\omega(n/2^{k/2}\cdot n^{1.5i\varepsilon})=\omega(n^{0.6+2i\varepsilon})$ cost.
Thus, the magnitude of each entry of $\alpha_1A^*_{j_1}+\alpha_2A^*_{j_2}+\cdots+\alpha_{k/2}A^*_{k/2}$ is $\Theta(n^{1.5i\varepsilon}).$
Thus, there exists $t\in[k/2],$ such that the absolute value of each entry of $\alpha_t A^*_{j_t}$ is at least $\Omega(n^{1.5i\varepsilon}/k).$
Then there are two cases.

The first case is $\group(t)<\group(j).$ Let $\group(t)=i'.$ Then $|\alpha_t|=\Omega(n^{1.5(i-i')\epsilon}/k).$
By Claim~\ref{cla:noise_cannot_fit_noise} again, with probability at least $1-(k/2)\cdot 2^{-\Theta(n/2^{k/2})},$ the size of
\begin{align*}
\{z\in[n-n^{0.1}]\mid \sign(u^*_z)=\sign(\alpha_1\Delta_{j_1,z})=\cdots=\sign(\alpha_{k/2}\Delta_{j_{k/2},z})=\sign(-\Delta_{j,z})\}
\end{align*}
is at least $\Omega(n/2^{k/2}).$ Thus, the cost to fit is at least $\Omega(n/2^{k/2})\cdot (n^{-0.2+ i' \epsilon}n^{1.5(i-i')\epsilon}/k)^2=\omega(n^{0.6+2i\varepsilon}).$

The second case is $\group(t)>\group(j).$ Let $\group(t)=i'.$  Then $|\alpha_t|=\Omega(n^{1.5(i-i')\epsilon}/k).$
By Claim~\ref{cla:noise_cannot_fit_noise} again, with probability at least $1-(k/2)\cdot 2^{-\Theta(n^{0.1}/2^{k/2})},$ the size of
\begin{align*}
\{z\in\{n-n^{0.1}+1,\cdots,n\}\mid \sign(u^*_z)=\sign(\alpha_1\Delta_{j_1,z})=\cdots=\sign(\alpha_{k/2}\Delta_{j_{k/2},z})=\sign(-\Delta_{j,z})\}
\end{align*}
is at least $\Omega(n^{0.1}/2^{k/2}).$ Thus, the fitting cost is at least $\Omega(n^{0.1}/2^{k/2})\cdot (n^{0.5+ 2i' \epsilon}n^{1.5(i-i')\epsilon}/k)=\omega(n^{0.6+2i\varepsilon}).$

By taking a union bound over all columns $j$, we have with probability at least $1-2^{n^{\Theta(1)}},$ the total cost to fit by a column subset $S$ is at least $\omega(kn^{1.6}).$

Then, by taking a union bound over all the $n\choose k$ number of sets $S,$ we complete the proof.
\end{proof}


\subsection{Column Subset Selection for the Reverse Huber Function}\label{sec:reverse_huber}
\begin{figure}[t]
\begin{center}
   \includegraphics[width=1.0\textwidth]{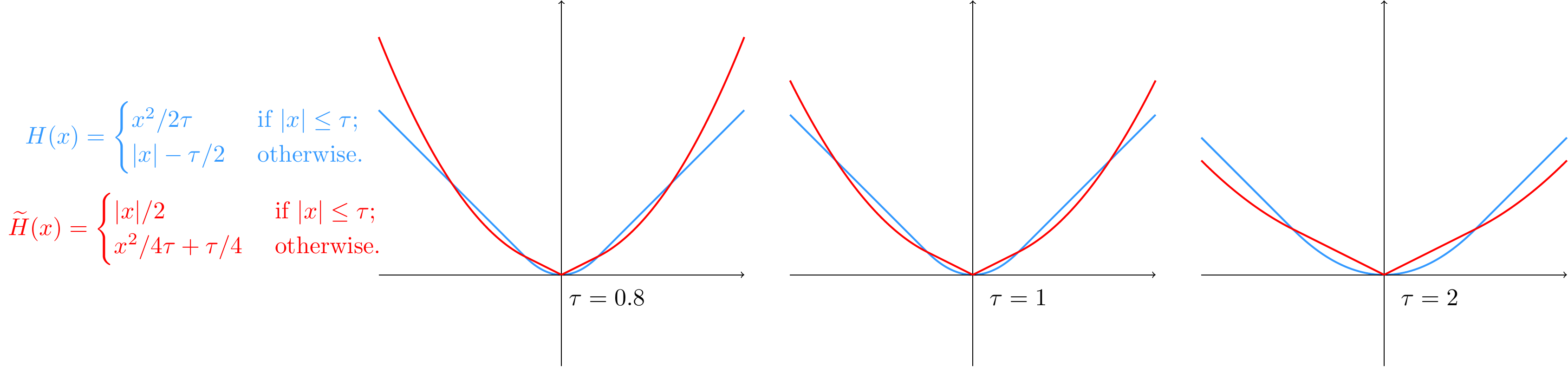}
   \caption{The blue curve is the Huber function which combines an $\ell_2$-like measure for small $x$ with an $\ell_1$-like measure for large $x$. The red curve is the ``reverse'' Huber function which combines an $\ell_1$-like measure for small $x$ with an $\ell_2$-like measure for large $x$.}
\end{center}
\end{figure}

In this section, we consider a ``reverse Huber function'': $g(x) = x^2$ if $x \geq 1$ and $g(x) = |x|$ for $x \leq 1$.

\begin{theorem}
Let $g(x)$ denote the ``reverse Huber function'' with $\tau =1$, i.e.,
\begin{align*}
H (x) =
\begin{cases}
x^2 /\tau, & \text{~if~} |x| > \tau;\\
|x|, & \text{~if~} |x| \leq \tau.
\end{cases}
\end{align*}
For any $n\geq 1,$ 
there is a matrix $A\in \R^{n \times n}$ such that, if we select 
only $1$ column to fit the entire matrix, there is 
no $n^{o(1)}$-approximation to the best rank-$1$ approximation,
 i.e., for any subset $S\subseteq[n]$ with $|S|=1,$
\begin{align*}
\min_{X \in \R^{|S| \times n} } \| A_S X -A \|_g \geq 
n^{\Omega(1)}\cdot\min_{\rank-1~A'} \| A' - A \|_g.
\end{align*}
\end{theorem}
\begin{proof}
Let $A \in \mathbb{R}^{n \times n}$ have one column that is $a = (n^{1/2}, 0, \ldots, 0)^\top$ and
$n-1$ columns that are each equal to $b = (0, 1/n, 1/n, \ldots, 1/n)^\top$. 
If we choose one column which has the form as $a$ to fit the other columns, the cost is at least $(n-1)^2/n=\Theta(n).$
If we choose one column which has the form as $b$ to fit the other columns, the cost is at least $(n^{1/2})^2=\Theta(n).$

Now we consider using a vector $c = (1/n^{1/4}, 1/n, 1/n, \ldots, 1/n)^\top$ to fit all the columns.
One can use $c$ to approximate
$a$ with cost at most $(n-1) \cdot n^{3/4}/n = \Theta(n^{3/4})$ by matching the first coordinate, while one
can use $c$ to approximate $b$ with cost at most $1/n^{1/4}$ by matching the last $n-1$ coordinates, and since
there are $n-1$ columns equal to $b$, the overall total cost of using $c$ to approximate matrix $A$ is
$\Theta(n^{3/4})$.

\end{proof}

\end{document}